\documentclass[10pt]{article}

\usepackage{cite} 
\usepackage{url}  
\usepackage{ifthen}  
\usepackage{multicol}   

\usepackage{amssymb}
\usepackage[english]{babel}
\usepackage{lmodern}
\usepackage{amsfonts}
\usepackage{amsmath}
\usepackage{amssymb}
\usepackage{color} 
\usepackage{url}
\usepackage{gastex}
\usepackage{algorithm2e}

\newcommand{\vpt}{v\mbox{\rm -Patterns}}
\newcommand{\ept}{e\mbox{\rm -Patterns}}
\newcommand{\ppt}{P_3\mbox{\rm -Patterns}}
\newcommand{\kpt}{K_3\mbox{\rm -Patterns}}
\newcommand{\Dpt}{\Delta\mbox{\rm -Patterns}}
\newcommand{\adj}{adj(\{p_1,p_2,p_3\})}
\newcommand{\acc}{\mbox{\rm acc-Motif}}

\newcounter{numerodoc}
\newcommand{\nextdoc}{\addtocounter{numerodoc}{1}\thenumerodoc}

\newcommand{\inset}{\delta^-}
\newcommand{\outset}{\delta^+}
\newcommand{\inoutset}{\delta^*}
\newcommand{\vizinhos}{\delta}

\newcommand{\mystar}{G^S}
\newcommand{\starv}{{V^S}}
\newcommand{\stare}{{E^S}}
\newcommand{\starc}{{v_c}}
\newcommand{\seta}{\mathcal{A}}
\newcommand{\setb}{\mathcal{B}}
\newcommand{\setc}{\mathcal{C}}
\newcommand{\sett}{\mathcal{T}}

\newcommand{\size}[1]{n_{#1}}
\newcommand{\m}[2]{{m_{#1,#2}}}
\newcommand{\ml}[2]{{m'_{#1,#2}}}
\newcommand{\mv}[3]{{m^{#1}_{#2,#3}}}
\newcommand{\mlv}[3]{{m'^{#1}_{#2,#3}}}

\newcommand{\pt}[2]{pattern(#1,#2)}
\newcommand{\Pt}{\mathcal{R}}

\newcommand{\pvai}[2]{pattern^{\rightarrow}(#1,#2)}
\newcommand{\pvolta}[2]{pattern^{\leftarrow}(#1,#2)}
\newcommand{\pvaivolta}[2]{pattern^{\leftrightarrow}(#1,#2)}
\newcommand{\ptaux}[3]{pattern(#1,#2,#3)}
\newcommand{\pvaiaux}[3]{pattern^{\rightarrow}(#1,#2,#3)}
\newcommand{\pvoltaaux}[3]{pattern^{\leftarrow}(#1,#2,#3)}
\newcommand{\pvaivoltaaux}[3]{pattern^{\leftrightarrow}(#1,#2,#3)}

\newcommand{\ptd}[2]{pattern'(#1,#2)}
\newcommand{\pvaid}[2]{pattern'^{\rightarrow}(#1,#2)}
\newcommand{\pvoltad}[2]{pattern'^{\leftarrow}(#1,#2)}
\newcommand{\pvaivoltad}[2]{pattern'^{\leftrightarrow}(#1,#2)}

\newtheorem{theorem}{Theorem}[section]

\newtheorem{corollary}[theorem]{Corollary}

\newtheorem{lemma}[theorem]{Lemma}

\newtheorem{problem}[theorem]{Problem}
\newtheorem{definition}[theorem]{Definition}
\newtheorem{fact}[theorem]{Fact}

\newcommand{\FIXME}[1]{\textcolor[rgb]{0,0,1}{#1}}

\definecolor{addcolor}{rgb}{0,0,1}

\begin{document}

\title{acc-Motif Detection Tool }

\author{L. A. A. Meira\\
\small Faculty of Technology, \\
\small University of Campinas, \\
\small Sao Pauo, Brazil
\and
V.  R. M\'{a}ximo\\
\small Institute of Science and Technology, \\
\small Federal University of S\~{a}o  Paulo, \\
\small Sao Pauo, Brazil\newline
\and
A. L. Fazenda\\
\small Institute of Science and Technology, \\
\small Federal University of S\~{a}o  Paulo, \\
\small Sao Pauo, Brazil
\and
A. F. da Concei\c{c}\~{a}o\\
\small Institute of Science and Technology, \\
\small Federal University of S\~{a}o  Paulo, \\
\small Sao Pauo, Brazil
}

\maketitle

\begin{abstract}
{\bf Background:} Network motif algorithms have been a topic of research mainly after the 2002-seminal paper from Milo \emph{et al}, that provided motifs as a way to uncover the basic building blocks of most networks. In Bioinformatics, motifs have been mainly applied in the field of gene regulation networks field.

{\bf Results:}
This paper proposes new  algorithms to exactly count isomorphic pattern motifs of sizes 3, 4 and 5 in directed graphs. 

Let $G(V,E)$ be a directed graph with $m=|E|$. We describe an $O({m\sqrt{m}})$ time complexity algorithm to count isomorphic patterns of size 3. In order to count isomorphic patterns of size 4, we propose an $O(m^2)$  algorithm. To count patterns with $5$ vertices, the algorithm is $O(m^2n)$.



{\bf Conclusion:}
The new algorithms were implemented and compared with FANMOD and Kavosh motif detection tools. The experiments
show that our algorithms are expressively faster than FANMOD and Kavosh's. We also let our motif-detecting tool available in the Internet.\\

{\bf keywords:} network motifs, complex networks, algorithm design and analysis, counting motifs, detecting motifs, motifs, discovery, motif isomorphism.
\end{abstract}

%

\section{Background}
\label{sec:intro}
Network Motifs, or simply motifs, correspond to small patterns that recurrently appear in a complex network~\cite{barabasi2003linked}. 
They can  be considered as the basic building blocks of complex networks and their understanding may be of interest in several areas, such as Bioinformatics \cite{MotifsGenoma,MotifsBio}, 
Communication \cite{MotifsWeb}, and Software Engineering \cite{MotifsSoftwareLivre}.

Finding network motifs has been a matter of attention mainly after the 2002-seminal paper from Milo \emph{et al}. \cite{MiloShen}, that proposed motifs as a way to uncover the structural design of complex networks. Nowadays, the design of efficient algorithms for network motif discovery is an up-to-date research area. Several surveys about motif detection algorithms were published in recent years \cite{ciriello2008review, Ribeiro:2009, Elisabeth2011}.

\subsection{Problem statement}

This paper formally  addresses the following three problems: 

\begin{problem}[Motifs-3]
\label{3motif} Given a directed graph $G(V,E)$, the problem Motifs-3 consists in counting  the number of connected induced subgraphs of G of size 3 grouped by the  13 isomorphic distinct graphs of size 3.
\end{problem}


\begin{problem}[Motifs-4]
\label{4motif} Given a directed graph $G(V,E)$, the problem Motifs-4 consists in counting  the number of connected induced subgraphs of G of size 4 grouped by the 199  isomorphic distinct graphs of size~4.
\end{problem}

\begin{problem}[Motifs-5]
\label{5motif} Given a directed graph $G(V,E)$, the problem Motifs-5 consists in counting  the number of connected induced subgraphs of G of size 5 grouped by the 9364  isomorphic distinct graphs of size~5.
\end{problem}


\subsection{Related work and tools}

The algorithms for motif detection can be based on two main approaches: exact counting or heuristic sampling. As these names might suggest, the former approach performs a precise count of the isomorphic pattern frequency. The latter uses statistics to estimate frequency value. Several exact search-based algorithms and tools can be found in the literature, such as MAVisto~\cite{schreiber2005mavisto}, NeMoFinder~\cite{Chen:2006}, Kavosh~\cite{kashani2009kavosh} and Grochow and Kellis~\cite{Grochow:2007}. Examples of sampling-based algorithms are MFinder~\cite{Kashtan:2004,Mfinder2002}, FANMOD~\cite{wernicke2006fanmod} and MODA~\cite{Moda2009}.

In 2010, Marcus and Shavitt~\cite{Marcus2010,Shavitt2012} provided an exactly algorithm $O(m^2)$ to count network motifs of size~4 in undirected graphs. In Section~\ref{sec:4u}, Figure~\ref{fig:correction} shows the only six connected isomorphic patterns with 4 vertices, that can be labeled as: tailed triangle, 4-cycle, 4-cycle with chord, 4-clique, 4-path and Claw. In fact, the paper provided six independent algorithms; each one devoted to count an undirected isomorphic pattern. Some ideas of Marcus and Shavitt are present in our approach. However, this paper provides a solution to directed graph cases. Furthermore, this work also solves 5-sized motifs.
A short version of this paper of this work appears in \cite{accMotifs}.

\subsection{Our approach: combinatorial acceleration}

This paper presents faster exact algorithms to Motif-3, Motif-4 and Motif-5 problems. Basically, two main  techniques are applied to improve computational complexity: first, our algorithm \emph{compute} the number of isomorphic patterns instead of listing induced subgraphs; second, our method does not actually check isomorphism. The
algorithms associate an integer variable with each isomorphic pattern and increment it directly.

Our algorithm was evaluated on transcription of biological networks (bacteria \emph{E. coli} and the yeast \emph{S. cerevisiae}) and public dataset networks with up to 13,000 vertices and 100,000 edges. 
The results are summarized on Tables \ref{temposk3}, \ref{temposk4} and \ref{temposk5}. Such  tables show a significant improvement in performance for Motif-3, Motif-4 and Motif-5. 
The $\acc$ was able to solve  instances
considered unfeasible in the current programs, which takes several
days to be solved even in probabilistic models.

We believe that the technique can be  extended for detecting motifs of higher sizes. 

The program was implemented in Java and it is made available as freeware in \url{http://www.ft.unicamp.br/~meira/accmotifs}. 


\subsection{Paper organization}

The remaining of this paper is organized as follows:
Section~\ref{sec:imple} describes the implemented algorithms; Section~\ref{sec:defs}  depicts the notation  used, Section~\ref{sec:size3}  introduces the new approach starting by the simple case
of counting isomorphic patterns of size 3, and Sections~\ref{sec:4u} and~\ref{subsec:directed4}
show the method applied to counting isomorphic patterns of size 4 in undirected and directed graphs, respectively.
Subsections \ref{subsec:undirected5} and \ref{subsec:directed5} are dedicated to count isomorphic patterns of size 5.
Section~\ref{sec:exp} presents the computational results, in comparison to other well-known tools available. Finally, Section~\ref{sec:conc} presents the conclusion and our view of future work.

\section{Implementation}
\label{sec:imple}

The existing exact algorithms to find network motifs are generally extremely costly in terms of CPU time and memory consumption, and present restrictions on the size of motifs \cite{kashani2009kavosh}. According to Cirielo and Guerra \cite{ciriello2008review}, motif algorithms typically consist of three steps: (i) listing connected subgraphs of $k$ vertices in the original graph
and in a set of randomized graphs;
(ii)  grouping them into isomorphic classes; and
(iii) determining  the statistical significance of isomorphic subgraph classes by comparing their frequencies to those of an
ensemble of random graphs. 
The core of this paper focuses on items (i) and (ii).

Section~\ref{sec:defs} presents the notation in use. Section~\ref{sec:size3} describes the algorithm to Motif-3 problem. Sections~\ref{sec:4u} and~\ref{subsec:directed4} describe  the algorithm to Motif-4 problem.
Sections~\ref{subsec:undirected5} and~\ref{subsec:directed5} describe  the algorithm to Motif-5 problem.


\subsection{Notation and definitions}
\label{sec:defs}

Let $G(V,E)$ be a directed graph with $n=|V|$ vertices and $m=|E|$ edges. Assuming that $m \geq n-1$. If $(u,v)\in E$ and $(v,u)\in E$, we say it is a bidirected edge. Alternatively, if only $(u,v)\in E$, we say it is a directed edge.

Given a vertex $v\in V$,  we partitioned  the neighbors of $v$ in three disjoint sets: $\inoutset(v)$, $\outset(v)$ and $\inset(v)$, as follows:

$$u\in\left \{ \begin{tabular}{ll}
	$\inoutset(v)$,&if $(u,v)\in E$ and $(v,u)\in E$ \\
	$\outset(v)$,&if $(v,u)\in E$ and $(u,v)\not\in E$ \\
	$\inset(v)$,&if $(u,v)\in E$ and $(v,u)\not\in E$ \\
\end{tabular}\right .$$


It means that $\inoutset(v)$ are the vertices with a bidirected edge to $v$.  The vertices with edges directed from $v$ are in $\outset(v)$ and
the vertices with edges directed to $v$ are in $\inset(v)$.

Sometimes, for convenience, we consider an undirected version of $G(V,E)$ called $G^*(V,E^*)$ where $\{u,v\}\in E^*$ if and only if $(u,v)\in E$ or $(v,u)\in E$, or both. Therefore,  we replace directed or bidirected edges of G by a single undirected edge in $G^*$.
Let us define $\vizinhos(v)$ as the neighbors of $v$, thus $u\in\vizinhos(v)$ if and only if $\{u,v\}\in E^*$.
Note that $\vizinhos(v)=\inoutset(v)\cup\outset(v)\cup\inset(v)$.


Given two disjoint sets $A\subseteq V$ and $B\subseteq V$, we define
$\inoutset(A,B)$ as the set of bidirected edges between the sets $A$ and $B$ and $\outset(A,B)$ as the directed edges from $A$ to $B$. We also define $\inoutset(A,A)$ for a single set $A$ as the bidirected edges $(u,v)$ with $u\in A$ and $v\in A$ and  
$\outset(A,A)$ as directed edges with $u\in A$ and $v\in A$.

We define the adjacency of a set of vertices $V'\subseteq V$ as $adj(V')=\left\{\cup_{v\in V'}\vizinhos(v)\right\}\setminus V'$.



\subsection{Counting isomorphic patterns of size 3}
\label{sec:size3}

To simplify notation usage, Table \ref{table:vars} defines a set of auxiliary variables related to a vertex $v$.


The symbol $ \seta^v=\inoutset(v)$ represents the set of bidirected neighbors of $v$, $ \setb^v=\outset(v)$ is the set of directed neighbors from $v$, and $ \setc^v=\inset(v)$ is the set of directed neighbors to $v$. The sets $\seta^v$, $\setb^v$ and $\setc^v$ define a partition of vertices in $v$ adjacency. For simplicity of notation, if the vertex $v$ is clear the superscript of $v$ can be removed.

The number of bidirected neighbors of $v$ is given by $n_a^v$. The number of directed edges from $\seta^v$ to $\setb^v$ is $\mv{v}{a}{b}$. The notation $m'$ is used to represent the number of bidirected edges, for example, $\mlv{v}{a}{b}$ is the number of bidirected 
edges between  $\seta$ and $\setb$ and ${m'}_{aa}$ is the number of bidirected edges inside $\seta^v$.

The algorithm to count 3-sized motifs is derived from Theorem~\ref{teo3}. However, in order to provide it with a better understanding, the following definition is needed:

\begin{definition}[$\vpt$]
Given a directed graph $G(V,E)$, we define $\vpt$, for any $v\in V$, as a set of induced subgraphs with three vertices, $\{v,x,y\}$, where $x$ and $y$ are in $\vizinhos(v)$, which means all induced subgraphs with
the vertex $v$ and more two vertices in its adjacency. 
The same definition is valid for the case of undirected graph.
\end{definition}

To illustrate the combinatorial optimization technique used in this paper, let us start by analyzing a simple case. Consider a star graph $\mystar(\starv,\stare)$ with center $\starc$ and neighbors $\seta^\starc$, $\setb^\starc$ and
$\setc^\starc$ as described in Figure \ref{fig:star}. 


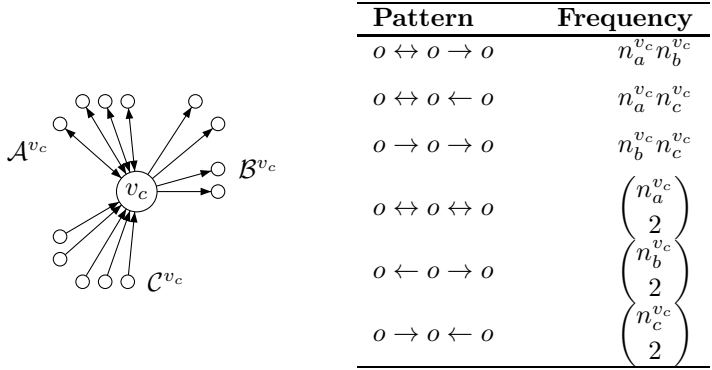
\begin{figure}[htb]
\begin{center}
\begin{tabular}{cc}
\setlength{\unitlength}{.6mm}
\begin{picture}(44,42)(10,20)
\node[Nh=3,Nw=3,Nmr=15](A2)(10,0){} 
\node[Nh=3,Nw=3,Nmr=15](A3)(15,0){} 
\node[Nh=3,Nw=3,Nmr=15](A4)(20,0){} 
\node[Nh=3,Nw=3,Nmr=15](B1)(5,5){} 
\node[Nh=3,Nw=3,Nmr=15](B2)(5,10){} 
\node[Nh=3,Nw=3,Nmr=15](B7)(5,35){} 
\node[Nh=3,Nw=3,Nmr=15](C4)(40,20){} 
\node[Nh=3,Nw=3,Nmr=15](C5)(40,25){} 
\node[Nh=3,Nw=3,Nmr=15](C7)(40,35){} 
\node[Nh=3,Nw=3,Nmr=15](D2)(10,40){} 
\node[Nh=3,Nw=3,Nmr=15](D3)(15,40){} 
\node[Nh=3,Nw=3,Nmr=15](D4)(20,40){} 
\node[Nh=3,Nw=3,Nmr=15](D7)(35,40){} 
\node[Nh=9,Nw=9,Nmr=15](C)(22,20){$v_c$} 
 \drawedge(A2,C){} 
 \drawedge(A3,C){} 
 \drawedge(A4,C){} 
 \drawedge(B1,C){} 
 \drawedge(B2,C){} 
 \drawedge[ATnb=1](C,B7){} 
 \drawedge(C,C4){} 
 \drawedge(C,C5){}  
 \drawedge(C,C7){}  
 \drawedge[ATnb=1](C,D2){} 
 \drawedge[ATnb=1](C,D3){} 
 \drawedge[ATnb=1](C,D4){} 
 \drawedge(C,D7){} 
    \nodelabel[ExtNL=y,NLangle=220,NLdist=3](B7){$\seta^\starc$}
  \nodelabel[ExtNL=y,NLangle=0,NLdist=3](C5){$\setb^\starc$}
    \nodelabel[ExtNL=y,NLangle=0,NLdist=3](A4){$\setc^\starc$}
\end{picture}
&~~~~
\begin{tabular}{lrr}
	\hline
	{\bf Pattern~~~~~}  & {\bf Frequency} \\
	\hline
	$\vspace{.2cm} o\leftrightarrow o \rightarrow o$  & $\displaystyle n_a^\starc n_b^\starc$\\
	$\vspace{.2cm} o\leftrightarrow o \leftarrow o$  & $\displaystyle n_a^\starc n_c^\starc$\\
	$\vspace{.2cm}o\rightarrow o \rightarrow o$  & $\displaystyle n_b^\starc n_c^\starc$\\
	$ o\leftrightarrow o \leftrightarrow o$  & $\displaystyle{{ n_a^\starc}\choose{2}}$\\
	$o\leftarrow o \rightarrow o$  & $\displaystyle{{ n_b^\starc}\choose{2}}$\\
	$o\rightarrow o \leftarrow o$  & $\displaystyle{{ n_c^\starc}\choose{2}}$\\
	\hline
\end{tabular}
\end{tabular}
\end{center}
\caption{Star graph and its sets.  Isomorphic pattern frequencies ont the right.}
\label{fig:star}
\setlength{\unitlength}{1mm}
\end{figure}

A naive algorithm to motif counting will compute all vertices in $\vizinhos(\starc)$ combined two by two. We argue that it is possible to compute the isomorphic patterns in $G^S$ in constant time $O(1)$, since we have precomputed the auxiliary variables of Table \ref{table:vars}.

Figure~\ref{fig:star} brings an insight about how our algorithm works. It shows, in a simple example, that it is possible to count isomorphic patterns without explicit listing all of them.
The right side of the figure depicts all possible patterns and occurrence frequencies in the star graph $G^S$.
It is possible to achieve, for instance, a number of exactly $n_a^\starc n_b^\starc$ occurrences of pattern ``$o\leftrightarrow o \rightarrow o$", which means a pattern with a center vertex linked to the left neighbor vertex by a bidirected edge and linked to the right neighbor by an edge directed to it. 


The algorithm to count isomorphic patterns in a general graph derives from the following theorem. 

\begin{theorem}\label{teo3}
Let $G(V,E)$ be a general graph and $v$ any vertex in $V$. The patterns occurrences in set $\vpt$
are given by Table~\ref{table:starb}.
\end{theorem}

\begin{proof}

This theorem is proved by induction. Observe that the $\vpt$ set considers the patterns containing $v$ and two vertices in  $\vizinhos(v)$.
Let $G'(E',V')$ be the graph $G$ induced by $v\cup\vizinhos(v)$.
The basic case is if the $G'$ is a star graph. In this case, the $\vpt$  frequencies are equal to Figure~\ref{fig:star} on the right. Table~\ref{table:starb} corresponds  to it if all $m^v$ variables are zero, which is the case in $G'$.

Suppose that a new $(x,y)$ directed edge is added to $G'$ where $x$ and $y$ are in $\vizinhos(v)$. 
The new graph has edges $E'\cup(u,v)$. At this moment, our sole interest is devoted to subgraph patterns that contain the vertex $v$. 

The number of pattern ``$ o \rightarrow o \rightarrow o$" in the original $G'$ is $n_b^v n_c^v$. 
If the new directed edge $(x,y)$ has $x\in\setc^v$ and $y\in\setb^v$,  one pattern ``$ o \rightarrow o \rightarrow o$ "
is removed and a cyclic pattern is added.  
The added pattern is shown in Table~\ref{table:starb}, Line 13.

If $u\in\seta$ and $v\in\seta$
one pattern ``$o\leftrightarrow o \leftrightarrow o$" is removed and another is added. The added pattern is 
shown in Table~\ref{table:starb},
 Line 11. 
For each edge added  in $\vizinhos(v)$, one pattern containing $v$ is removed
and another pattern containing $v$ is added. 

A straightforward generalization is observed for an arbitrary number of added edges. Suppose $\mv{v}{c}{b}$
directed edges are added from $\setc^v$ to $\setb^v$. The number of
``$ o\rightarrow  o\rightarrow o$" decreases $\mv{v}{c}{b}$ units and  exactly $\mv{v}{c}{b}$
occurrences arise from a new one, which can be seen in Table~\ref{table:starb}, Line 13.  
\end{proof}

Thus, given a graph $G(V,E)$, for each vertex $v\in V$,
it is possible to obtain the $\vpt$ frequencies using Theorem~\ref{teo3}.
Table \ref{table:starb} shows this pattern frequency (see variable definition in Table \ref{table:vars}).

If the variables of Table \ref{table:vars} were preprocessed, it is possible to
calculate all isomorphic patterns of size $3$  containing a vertex $v\in V$  and two other neighbors of $v$
in $O(1)$.

The pattern containing $v$, a neighbor of $v$, and a non-neighbor of $v$ 
will be ignored by the $\vpt$ set.
Fortunately, valid patterns involving these vertices could be computed by their center vertex at another moment.
Patterns related to $C_3$ will be considered three times each. The pattern $C_3$ is in the $\vpt$ set of vertices $v_1,v_2$ and $v_3$ in $C_3$. Therefore, a simple correction must be applied. The final counter of $C_3$ related isomorphic patterns must be divided by three to provide the correct value.



The Algorithm~\ref{alg:s3} counts motifs  patterns of size 3. In fact, the algorithm does not perform any isomorphic matching.
The algorithm creates a vector $h$ with thirteen integers and initializes them with zero. In this vector, the pattern ``$o\leftrightarrow o \rightarrow o$" is arbitrarily associated with $h[0]$, the pattern ``$o\leftrightarrow o \leftarrow o$" is arbitrarily  associated with $h[1]$, ``$o\rightarrow o \rightarrow o$" with $h[2]$, and so on. The algorithm computes Table~\ref{table:starb} frequencies  for each $v \in V$. The frequencies of Table~\ref{table:starb} are incremented in vector $h$ directly. The algorithm output is vector $h$, containing thirteen isomorphic pattern frequencies.

\begin{algorithm}[htb]
\LinesNumbered
\SetAlgoLined
\KwIn{Directed graph $G(V,E)$}
\KwOut{ Histogram to 13 isomorphic patterns to motifs of size 3}
Create a histogram data structure to count isomorphic patterns\\
Calculate the variables of Table \ref{table:vars} to all vertices.\\
\ForEach{$v\in V$}{
Calculate the number of patterns involving vertex $v$ using frequencies of Table \ref{table:starb}.\\
For each pattern, add this frequency counter to histogram. 
}
\If{\rm The undirected version of the pattern is the cycle graph $C_3$}{
Divide the frequency counter by 3
}
\Return{\rm The histogram.}
\caption{ Count 3 Sized Patterns Algorithm.}
\label{alg:s3}
\end{algorithm}




The complexity of the algorithm is dominated by Line 2, which computes variables in Table~\ref{table:vars}. All operations in Algorithm~\ref{alg:s3}, except Line 2, are $O(\FIXME{n})$. The next section shows how to compute Line 2  in $O(m\sqrt{m})$.

\subsection{Preprocessing Table \ref{table:vars}}
\label{sec:preproc}

This section argues that, given a directed graph $G(V,E)$, 
it is possible to compute Table \ref{table:vars} sets and variables  in $O(a(G)m)$, where 
$a(G)$ is the arboricity of the undirected version of $G$. Arboricity was introduced by Nash-Williams~\cite{nash1964decomposition}, the arboricity $a(G)$ of a graph $G$ is the minimum number of forests into which its edges can be partitioned. It is known \cite{Chiba:1985} that $a(G)=O(\sqrt{E})$ to any graph, so the execution complexity is also $O(m\sqrt{m})$. 

First, for each vertex $v\in V$, create three sets  $\seta^v,\setb^v$, and $\setc^v$. 
Algorithm~\ref{alg:table1a} describes how to compute such variables in $O(m)$.

\begin{algorithm}[htb]
\LinesNumbered
\SetAlgoLined
\KwIn{Directed graph $G(V,E)$}
\KwOut{ Variables $\seta^v,\setb^v$, and $\setc^v$ and $n^v_a$, $n^v_b$ and $n^v_c$ for all $v\in V$}
\ForEach{$v\in V$}{
 $(\seta^v,\setb^v,\setc^v)\gets(\emptyset,\emptyset,\emptyset)$\\
}
\ForEach{bidirected $(u,v)\in E$}{
    $\seta^u\gets\seta^u\cup\{v\}$\\
    $\seta^v\gets\seta^v\cup\{u\}$    
}
\ForEach{directed $(u,v)\in E$}{
    $\setb^u\gets\setb^u\cup\{v\}$\\
    $\setc^v\gets\setc^v\cup\{u\}$    
}
\ForEach{$v\in V$}{
    $(n^v_a,n^v_b,n^v_c)\gets (|\seta^v|,|\setb^v|,|\setc^v|)$
}
\caption{Create $\{\seta^v,\setb^v,\setc^v\}$   variables.}
\label{alg:table1a}
\end{algorithm}

Algorithm~\ref{alg:alternative}  computes variables~$\{\mv{v}{a}{a},\ldots,\mlv{v}{c}{c}\}$.
The complexity is dominated by Line~5, the algorithm that lists all triangles in an undirected graph. 
If Chiba and Nishizek algorithm~\cite{Chiba:1985} is used to list all triangles, an $O(a(G)m)$ algorithm is obtained, where $a(G)$ is the arboricity of $G$.

\begin{algorithm}[htb]
\LinesNumbered
\SetAlgoLined
\KwIn{Directed graph $G(V,E)$}
\KwOut{ Variables  $\{\mv{v}{a}{a},\ldots,\mlv{v}{c}{c}\}$.}
Let $G^*(V,E^*)$ be the undirected version of $G(V,E)$.\\
\ForEach{$v\in V$}{
All variables in $\{\mv{v}{a}{a},\ldots,\mlv{v}{c}{c}\}$ start with zero.\\
}
List all triangles of $G^*(V,E^*)$ and save in variable $\sett$\\
\ForEach{triangle $(v_1,v_2,v_3)\in \sett$}{
  Let $(v,x,y)\gets(v_1,v_2,v_3)$
\begin{tabular}{lcc}
	\hline
	\bf{If vertices~~~~~}  &  \bf{Var to increment} & \bf{Var to increment}\\
	 & if $(x,y)\in E$ is bidirect~~~ & If $(x,y)\in E$ is direct~~ \\ 	\hline
	 $x \in \seta^v$ and $y\in \seta^v$ & $\mlv{v}{a}{a}$  & $\mv{v}{a}{a}$\\
	$x \in \seta^v$ and $y\in \setb^v$ & $\mlv{v}{a}{b}$ and $\mlv{v}{b}{a}$&$\mv{v}{a}{b}$\\
	$x \in \seta^v$ and $y\in \setc^v$ & $\mlv{v}{a}{c}$ and $\mlv{v}{c}{a}$&$\mv{v}{a}{c}$\\
	$x \in \setb^v$ and $y\in \seta^v$ & $\mlv{v}{b}{a}$ and $\mlv{v}{a}{b}$&$\mv{v}{b}{a}$\\
	$x \in \setb^v$ and $y\in \setb^v$ & $\mlv{v}{b}{b}$&$\mv{v}{b}{b}$\\
	$x \in \setb^v$ and $y\in \setc^v$ & $\mlv{v}{b}{c}$ and $\mlv{v}{c}{b}$&$\mv{v}{b}{c}$\\
	$x \in \setc^v$ and $y\in \seta^v$ & $\mlv{v}{c}{a}$ and $\mlv{v}{a}{c}$&$\mv{v}{c}{a}$\\
	$x \in \setc^v$ and $y\in \setb^v$ & $\mlv{v}{c}{b}$ and $\mlv{v}{b}{c}$& $\mv{v}{c}{b}$\\
	$x \in \setc^v$ and $y\in \setc^v$ & $\mlv{v}{c}{c}$&$\mlv{v}{c}{c}$\\
	\hline 
	~\\
\end{tabular}\\
Do the same to $(v,x,y)\gets(v_2,v_1,v_3)$ and to $(v,x,y)\gets(v_3,v_1,v_2)$
}
\caption{Create $\{\mv{v}{a}{a},\ldots,\mlv{v}{c}{c}\}$   variables.}
\label{alg:alternative}
\end{algorithm}

%
%
%

It is possible to notice that, in essence, Algorithm~\ref{alg:alternative} is processing all triangles
in $G(V,E)$. Each increment in $\mv{v}{a}{b}$  is an operation in a triangle containing $v$ and two connected
neighbors. We remark the existence of more straightforward implementations of Algorithm~\ref{alg:alternative}, but
the use of Chiba and Nishizek algorithm \cite{Chiba:1985} to list all triangles as a subroutine simplifies the
complexity analysis.

Thus, it is possible to conclude that the Algorithm~\ref{alg:s3}, which solves the Motifs-3 problem, presents an $O(a(G)m)$ time complexity. The memory used in the algorithm is linear in relation to the memory used to represent G(V,E).

\subsection{Counting isomorphic patterns of size 4 in \emph{undirected} graphs}
\label{sec:4u}

To show our solution of  Motif-4 problem, let us start with an undirected version of the problem.
The directed case involves more details and will be considered in Section \ref{subsec:directed4}. 

Similarly to the previous section, the following definition needs to be known beforehand:

\begin{definition}[$\ept$]
\label{ept}
Given a directed graph $G(V,E)$, we define an $\ept$, for any $e\in E$, as a set of patterns with four vertices, $\{u,v,v_1,v_2\}$, where $(u,v)=e$ and $v_1$ and $v_2$ are in $adj(\{u,v\})$. The $\ept$ set has patterns with
the edge $e$ and more two vertices in its adjacency. 
The same applies to the undirected graph.
\end{definition}

The approach to count patterns with four vertices is countingt $\ept$ to all $e\in E$. 
Let us define $Z^e=\vizinhos(u)\cap\vizinhos(v)\setminus\{u,v\}$ as the vertices adjacent to $u$ and $v$,
$X^e=\vizinhos(u)\setminus (Z^e\cup\{v\})$ as the vertices only in  $u$ adjacency, and $Y^e=\vizinhos(v)\setminus (Z^e\cup \{u\})$ 
as the vertices only in  $v$ adjacency. If the edge $e$ is clear it can be omitted from the superscript. We define  $\size{x}^e$,
$\size{y}^e$ and $\size{z}^e$ as the sizes $|X^e|$, $|Y^e|$ and $|Z^e|$, respectively. See Figure \ref{fig:stare}.

The $C_k$ is the cycle graph with $k$ vertices, the $S_k$ is the star graph with a center and $k$ leaves, and
the $K_k$ is the complete graph of size $k$. The $K_k\setminus\{e\}$ is the complete $K_k$ graph without 
an arbitrary edge $e$. The $P_k$ is the 
path graph with $k$ vertices.

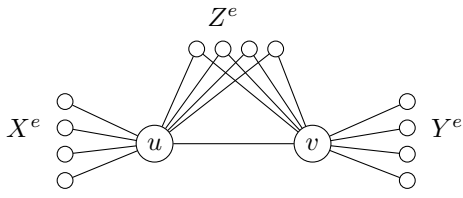
\begin{figure}[htb]
\centerline{
\begin{tabular}{cc}
\setlength{\unitlength}{.7mm}
\begin{picture}(80,35)(0,10)
\node[Nh=3,Nw=3,Nmr=15](B1)(5,0){} 
\node[Nh=3,Nw=3,Nmr=15](B2)(5,5){} 
\node[Nh=3,Nw=3,Nmr=15](B3)(5,10){} 
\node[Nh=3,Nw=3,Nmr=15](B4)(5,15){} 
\node[Nh=3,Nw=3,Nmr=15](C3)(70,0){} 
\node[Nh=3,Nw=3,Nmr=15](C4)(70,5){} 
\node[Nh=3,Nw=3,Nmr=15](C5)(70,10){} 
\node[Nh=3,Nw=3,Nmr=15](C6)(70,15){} 
\node[Nh=3,Nw=3,Nmr=15](D1)(30,25){} 
\node[Nh=3,Nw=3,Nmr=15](D2)(45,25){} 
\node[Nh=3,Nw=3,Nmr=15](D3)(35,25){} 
\node[Nh=3,Nw=3,Nmr=15](D4)(40,25){} 
\node[Nh=7,Nw=7,Nmr=15](C)(22,7){$u$} 
\node[Nh=7,Nw=7,Nmr=15](D)(52,7){$v$} 
 \drawedge[ATnb=0,AHnb=0](B1,C){} 
 \drawedge[ATnb=0,AHnb=0](B2,C){} 
 \drawedge[ATnb=0,AHnb=0](B3,C){} 
 \drawedge[ATnb=0,AHnb=0](C,B4){} 
 \drawedge[ATnb=0,AHnb=0](D,C3){} 
 \drawedge[ATnb=0,AHnb=0](D,C4){} 
 \drawedge[ATnb=0,AHnb=0](D,C5){} 
 \drawedge[ATnb=0,AHnb=0](D,C6){} 
 \drawedge[ATnb=0,AHnb=0](C,D1){} 
 \drawedge[ATnb=0,AHnb=0](C,D2){} 
 \drawedge[ATnb=0,AHnb=0](C,D3){} 
 \drawedge[ATnb=0,AHnb=0](C,D4){} 
 \drawedge[ATnb=0,AHnb=0](D,D1){} 
 \drawedge[ATnb=0,AHnb=0](D,D2){} 
 \drawedge[ATnb=0,AHnb=0](D,D3){} 
 \drawedge[ATnb=0,AHnb=0](D,D4){} 
 \drawedge[ATnb=0,AHnb=0](C,D){}  
    \nodelabel[ExtNL=y,NLangle=180,NLdist=3](B3){$X^e$}
  \nodelabel[ExtNL=y,NLangle=0,NLdist=3](C5){$Y^e$}
    \nodelabel[ExtNL=y,NLangle=90,NLdist=3](D3){$Z^e$}
    \end{picture}
& ~~~~~~
\begin{tabular}{lr}
	\hline
	\textbf{Pattern~~~~~}  & {\bf Frequency} \\
	\hline
	$\vspace{.2cm} P_4$  & $\displaystyle \size{x} \size{y}$\\
	$\vspace{.2cm} x \leftrightarrow C_3$  & $\displaystyle (\size{x} +\size{y}) \size{z}$\\
	$ S_3$  & $\displaystyle {\size{x}\choose{2}}+{\size{y}\choose{2}}$\\
	$K_4\setminus\{e\}$  & $\displaystyle {\size{z}\choose{2}}$\\
	\hline
\end{tabular}
\end{tabular}
}
\caption{Sets associated with $e=\{u,v\}$ and $\ept$ frequency of the graph on the left.}
\label{fig:stare}
\setlength{\unitlength}{1mm}
\end{figure}

The algorithm to count isomorphic patterns derives from the following theorem. 

\begin{theorem}\label{teo4u}
Let $G(V,E)$ be a general undirected graph and $e=\{u,v\}$ any edge in $E$. The patterns occurrences in set $\ept$
is given by Table~\ref{tab:ept}.
\end{theorem}

\begin{proof}
This theorem is also proved by induction in the number of edges. 
Let $G'$ be the graph $G$ induced by $\{u,v\}\cup adj(\{u,v\}).$
The basic case is if $G'$ is the graph in Figure~\ref{fig:stare}. In this case, the $\ept$  frequencies are equal to Figure~\ref{fig:stare} on the right. Table~\ref{tab:ept} corresponds to it if all $m^e$ variables are zero, which is true for the graph at issue.


%

If one edge $e'$ is added into $X^e$, one $\ept$ $S_3$  has to be replaced by one $\ept$ 
$x\leftrightarrow C_3$. The same applies to $Y^e$. If one edge is added to $Z^e$, one
$\ept$ $K_4\setminus\{e\}$  has to be replaced by one $\ept$ $K_4$. If one edge is added between $X^e$ and $Y^e$, 
one $\ept$ $P_4$  needs to be removed and one $\ept$ $C_4$ must be added. If one edge is added between $X^e$ or $Y^e$ to $Z^e$,
 one $\ept$ $x \leftrightarrow C_3$ has to be deleted and one $\ept$ $K_4\setminus\{e\}$ must be added. 
 Let $\mv{e}{x}{y}$ be the number of edges between sets $X^e$ and $Y^e$. Similarly consider variables $\mv{e}{x}{x}$, $\mv{e}{x}{z}$, $\mv{e}{y}{y}$, $\mv{e}{y}{z}$, and $\mv{e}{z}{z}$.
 Thus, Table \ref{tab:ept} presents the  $\ept$ frequency to $e=\{u,v\}$. 
 \end{proof}
 

 The algorithm to count isomorphic patterns of size four will sum up the   $\ept$ frequencies for all $e\in E$.
 
 Similarly to Section \ref{sec:size3}, the pattern
containing $\{u,v\}$, a neighbor of $u$ or $v$ and a non-neighbor of $u$ and $v$ is not considered in $\ept$ set.  It applies to pattern $P_4$  and a
non-central edge and pattern $x \leftrightarrow C_3$ and one edge in  $C_3$.
 Fortunately, these patterns will be considered later for other edges. 
An induced subgraph pattern can be considered in distinct $\ept$ sets. So,
the final histogram needs a small correction. If a pattern appears at the $\ept$ set
for $a$ of its edges, the number of occurrences in the final histogram must be divided by $a$.

The following fact describe this situation.

\begin{fact}
\label{fact:correction}
Based on Definition \ref{ept} and Figure~\ref{fig:correction}, 
the $C_4$ patterns are considered by four edges. It is in $\ept$ for each of its four edges. The $P_4$ 
is in the $\ept$ set only for the central edge. The $S_3$ is in $\ept$ for each of its three edges. The tailed 
triangle is in $\ept$ in three of its four edges. The $K_4\setminus\{e\}$ is in $\ept$ for each of its five edges and
$K_4$ is in $\ept$ in each of its six edges.
\end{fact}

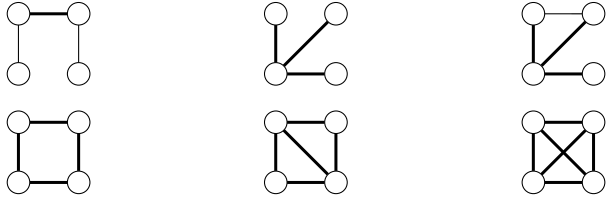
\begin{figure}
\centerline{
\setlength{\unitlength}{2mm}
\begin{tabular}{ccc} 
\hspace{3cm}  &  \hspace{3cm}  & \hspace{3cm} \\ 
\begin{picture}(0,0)(6,6)\node[Nh=1.5,Nw=1.5,Nmr=1.5](0)(4,5){}\node[Nh=1.5,Nw=1.5,Nmr=1.5](1)(8,5){}\node[Nh=1.5,Nw=1.5,Nmr=1.5](2)(4,9){}\node[Nh=1.5,Nw=1.5,Nmr=1.5](3)(8,9){}\drawedge[ATnb=0,AHnb=0](0,2){}\drawedge[ATnb=0,AHLength=0,ATLength=1, linewidth=.2](2,3){}\drawedge[ATnb=0,AHnb=0](3,1){}
\end{picture} &
\begin{picture}(0,0)(6,6)\node[Nh=1.5,Nw=1.5,Nmr=1.5](0)(4,5){}\node[Nh=1.5,Nw=1.5,Nmr=1.5](1)(8,5){}\node[Nh=1.5,Nw=1.5,Nmr=1.5](2)(4,9){}\node[Nh=1.5,Nw=1.5,Nmr=1.5](3)(8,9){}\drawedge[ATnb=0,AHnb=0,linewidth=.2](0,2){}\drawedge[ATnb=0,AHnb=0,linewidth=.2](3,0){}\drawedge[ATnb=0,AHnb=0,linewidth=.2](1,0){}
\end{picture}
  &
\begin{picture}(0,0)(6,6)\node[Nh=1.5,Nw=1.5,Nmr=1.5](0)(4,5){}\node[Nh=1.5,Nw=1.5,Nmr=1.5](1)(8,5){}\node[Nh=1.5,Nw=1.5,Nmr=1.5](2)(4,9){}\node[Nh=1.5,Nw=1.5,Nmr=1.5](3)(8,9){}\drawedge[ATnb=0,AHnb=0,linewidth=.2](0,2){}\drawedge[ATnb=0,AHLength=0,ATLength=1](2,3){}\drawedge[ATnb=0,AHnb=0,linewidth=.2](3,0){}\drawedge[ATnb=0,AHnb=0,linewidth=.2](1,0){}
\end{picture} \\
~\vspace{.6cm}\\
\begin{picture}(0,0)(6,6)\node[Nh=1.5,Nw=1.5,Nmr=1.5](0)(4,5){}\node[Nh=1.5,Nw=1.5,Nmr=1.5](1)(8,5){}\node[Nh=1.5,Nw=1.5,Nmr=1.5](2)(4,9){}\node[Nh=1.5,Nw=1.5,Nmr=1.5](3)(8,9){}\drawedge[ATnb=0,AHnb=0,linewidth=.2](0,2){}\drawedge[ATnb=0,AHLength=0,ATLength=1, linewidth=.2](2,3){}\drawedge[ATnb=0,AHnb=0,linewidth=.2](3,1){}\drawedge[ATnb=0,AHnb=0,linewidth=.2](1,0){}
\end{picture} &
\begin{picture}(0,0)(6,6)\node[Nh=1.5,Nw=1.5,Nmr=1.5](0)(4,5){}\node[Nh=1.5,Nw=1.5,Nmr=1.5](1)(8,5){}\node[Nh=1.5,Nw=1.5,Nmr=1.5](2)(4,9){}\node[Nh=1.5,Nw=1.5,Nmr=1.5](3)(8,9){}\drawedge[ATnb=0,AHnb=0,linewidth=.2](0,2){}\drawedge[ATnb=0,AHLength=0,ATLength=1, linewidth=.2](2,3){}\drawedge[ATnb=0,AHnb=0,linewidth=.2](3,1){}\drawedge[ATnb=0,AHnb=0,linewidth=.2](1,0){}\drawedge[ATnb=0,AHnb=0,linewidth=.2](2,1){}
\end{picture} &
\begin{picture}(0,0)(6,6)\node[Nh=1.5,Nw=1.5,Nmr=1.5](0)(4,5){}\node[Nh=1.5,Nw=1.5,Nmr=1.5](1)(8,5){}\node[Nh=1.5,Nw=1.5,Nmr=1.5](2)(4,9){}\node[Nh=1.5,Nw=1.5,Nmr=1.5](3)(8,9){}\drawedge[ATnb=0,AHnb=0,linewidth=.2](0,2){}\drawedge[ATnb=0,AHLength=0,ATLength=1, linewidth=.2](2,3){}\drawedge[ATnb=0,AHnb=0,linewidth=.2](3,1){}\drawedge[ATnb=0,AHnb=0,linewidth=.2](1,0){}\drawedge[ATnb=0,AHnb=0,linewidth=.2](2,1){} \drawedge[ATnb=0,AHnb=0,linewidth=.2](3,0){}
\end{picture}
\\
~
\end{tabular}
}
\caption{Patterns and its edges. The pattern is in $\ept$  of edges in bold. See Definition~\ref{ept}.}
\label{fig:correction}
\end{figure}


\begin{algorithm}[htb]
\LinesNumbered
\SetAlgoLined
\KwIn{Undirected graph $G(V,E)$}
\KwOut{ Histogram to 6 isomorphic patterns to motifs of size 4}
Create a histogram to count isomorphic patterns\\
\ForEach{$e\in E$}{
Calculate  variables $n^e_x,n^e_y,n^e_z,\mv{e}{x}{x},\mv{e}{x}{y},\mv{e}{x}{z},\mv{e}{y}{y},\mv{e}{y}{z},\mv{e}{z}{z}$.\\
 Calculate the frequency of $\ept$ using Table \ref{tab:ept}.\\
  For each pattern, add its frequency counter to histogram.\\
}
{\bf if  \rm  The histogram pattern is (See Fact~\ref{fact:correction}.):}\\
 \Indp {$x\leftrightarrow C_3$ or $S_3$: Divide the frequency counter by 3}\\
  {$C_4$:  Divide the frequency counter by 4}\\
  {$K_4\setminus\{e\}$: Divide the frequency counter by 5}\\
  {$K_4$: Divide the frequency counter by 6}\\
 \Indm   \Return{\rm The histogram.}
\caption{   Count 4 Sized Patterns Algorithm.}
\label{alg:und4}
\end{algorithm}

The Algorithm~\ref{alg:und4} counts 4-sized subgraphs grouped by isomorphic patterns. The complexity is dominated by Line 3, the time to 
calculate the needed variables. All other steps are $O(m)$. 

As in the previous section, there is no isomorphism detecting algorithm.
The histogram is represented by a vector $h$ with 6 positions.  The algorithm associates each pattern 
with an arbitrary \emph{hard coded} position in the vector. For instance, 
 patterns $(P_4,x\leftrightarrow C_3,S_3,C_4,K_4\setminus\{e\},K_4)$ may be related with $(h[0],h[1],h[2],h[3],h[4],h[5])$, respectively.
To sum the $\ept$ occurrences for a specific $e\in E$ it is sufficient to update the integer variables in the 
vector $h$ using Table~\ref{tab:ept} rule. The algorithm output is the histogram vector containing pattern frequencies.

The Algorithm~\ref{alg:vars} computes the needed variables, according to Line 3 of Algorithm~\ref{alg:und4}. 
The algorithm
has an $O(m)$ complexity for each $e\in E$. The variables  $n^e_{x}$, $n^e_{y}$ and $n^e_{z}$ 
are simpler so their calculus was omitted. Note that checking whether $x\in Z^e$ for an arbitrary vertex $x\in V$ and  edge $e=\{u,v\}$
is equivalent to checking whether $(x,u)\in E$ and $(x,v)\in E$, and can be performed in $O(1)$. Checking whether $x\in X^e$ and $x\in Y^e$  is a similar procedure.


\begin{algorithm}[htb]
\LinesNumbered
\SetAlgoLined
\KwIn{Undirected graph $G(V,E)$ and an edges $e\in E$}
\KwOut{Variables   $\{\mv{e}{x}{x},\ldots,\mv{e}{z}{z}\}$.}
 All variables in $\{\mv{e}{x}{x},\ldots,\mv{e}{z}{z}\}$ start with zero.\\
\ForEach{$(x,y)\in E$ }{
\begin{tabular}{lc}
	\hline
	\bf{If vertices~~~~~}  &  \bf{Var to increment} \\\hline
	 $x \in X^e$ and $y\in X^e$ & $\mv{e}{x}{x}$  \\
	 $x \in X^e$ and $y\in Y^e$ & $\mv{e}{x}{y}$\\
	 $x \in X^e$ and $y\in Z^e$ & $\mv{e}{x}{z}$\\
	 $x \in Y^e$ and $y\in Y^e$ & $\mv{e}{y}{y}$\\
	 $x \in Y^e$ and $y\in Z^e$ & $\mv{e}{y}{z}$\\
	 $x \in Z^e$ and $y\in Z^e$ & $\mv{e}{z}{z}$\\	 
	\hline 
	~\\
\end{tabular}
}
\Return{ \rm The computed variables.}
\caption{ Create $\{\mv{e}{x}{x},\ldots,\mv{e}{z}{z}\}$   variables for a given $e\in E$.}
\label{alg:vars}
\end{algorithm}

We can conclude that 
 Algorithm~\ref{alg:und4} counts 4-sized subgraphs grouped by isomorphic patterns
 in $O(m^2)$
in an undirected graph $G(V,E)$. Moreover, the additional memory to store the variables is $\Theta(m)$

\subsection{Counting isomorphic patterns of size 4 in \emph{directed} graphs}
\label{subsec:directed4}

No new concept is needed to extend the previous algorithm to the directed version. However, a large number of sets and  variables have to be dealt with. Variables and sets related to an edge $e$ are presented next.

Considering an edge $e=(u,v)$, it is possible to define 15 sets associated with it. 
$$\sett^e=\{A^e_1,B^e_1,C^e_1,A^e_2,B^e_2,C^e_2,AA^e,AB^e,AC^e,BA^e,BB^e,BC^e,CA^e,CB^e,CC^e\}$$ defined as follows (see Figure \ref{fig:edgesets}): $AA^e\gets \seta^u\cap\seta^v$, $AB^e\gets \seta^u\cap\setb^v$, $AC^e\gets \seta^u\cap\setc^v$, $BA^e\gets \setb^u\cap\seta^v$, $BB^e\gets \setb^u\cap\setb^v$, $BC^e\gets \setb^u\cap\setc^v$, $CA^e\gets \setc^u\cap\seta^v$, $CB^e\gets \setc^u\cap\setb^v$, and $CC^e\gets \setc^u\cap\setc^v$.

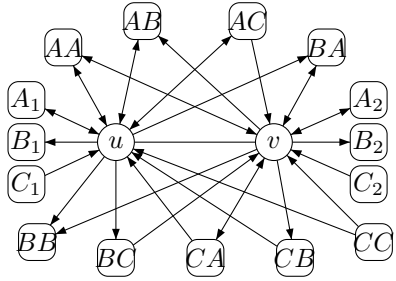
\begin{figure}[htb]
\begin{center}
\setlength{\unitlength}{.7mm}
\begin{picture}(80,35)(0,10)
\node[Nh=7,Nw=7,Nmr=2](C1)(5,-1){$C_1$} 
\node[Nh=7,Nw=7,Nmr=2](B1)(5,7){$B_1$} 
\node[Nh=7,Nw=7,Nmr=2](A1)(5,15){$A_1$} 
\node[Nh=7,Nw=7,Nmr=2](C2)(70,-1){$C_2$} 
\node[Nh=7,Nw=7,Nmr=2](B2)(70,7){$B_2$} 
\node[Nh=7,Nw=7,Nmr=2](A2)(70,15){$A_2$} 
\node[Nh=7,Nw=7,Nmr=2](AA)(12,25){$AA$} 
\node[Nh=7,Nw=7,Nmr=2](AB)(27,30){$AB$} 
\node[Nh=7,Nw=7,Nmr=2](AC)(47,30){$AC$} 
\node[Nh=7,Nw=7,Nmr=2](BA)(62,25){$BA$} 
\node[Nh=7,Nw=7,Nmr=2](BB)(7,-12){$BB$} 
\node[Nh=7,Nw=7,Nmr=2](BC)(22,-15){$BC$} 
\node[Nh=7,Nw=7,Nmr=2](CA)(39,-15){$CA$} 
\node[Nh=7,Nw=7,Nmr=2](CB)(56,-15){$CB$}
\node[Nh=7,Nw=7,Nmr=2](CC)(71,-12){$CC$}  
\node[Nh=7,Nw=7,Nmr=15](C)(22,7){$u$} 
\node[Nh=7,Nw=7,Nmr=15](D)(52,7){$v$} 
 \drawedge[ATnb=1,AHnb=1](A1,C){} 
 \drawedge[ATnb=1,AHnb=0](B1,C){} 
 \drawedge[ATnb=0,AHnb=1](C1,C){} 
 \drawedge[ATnb=1,AHnb=1](D,A2){} 
 \drawedge[ATnb=0,AHnb=1](D,B2){} 
 \drawedge[ATnb=1,AHnb=0](D,C2){} 
 \drawedge[ATnb=1,AHnb=1](C,AA){} 
 \drawedge[ATnb=1,AHnb=1](C,AB){} 
 \drawedge[ATnb=1,AHnb=1](C,AC){} 
 \drawedge[ATnb=0,AHnb=1](C,BA){} 
 \drawedge[ATnb=0,AHnb=1](C,BB){} 
 \drawedge[ATnb=0,AHnb=1](C,BC){} 
 \drawedge[ATnb=1,AHnb=0](C,CA){} 
 \drawedge[ATnb=1,AHnb=0](C,CB){} 
 \drawedge[ATnb=1,AHnb=0](C,CC){}
  \drawedge[ATnb=1,AHnb=1](D,AA){} 
 \drawedge[ATnb=0,AHnb=1](D,AB){} 
 \drawedge[ATnb=1,AHnb=0](D,AC){} 
 \drawedge[ATnb=1,AHnb=1](D,BA){} 
 \drawedge[ATnb=0,AHnb=1](D,BB){} 
 \drawedge[ATnb=1,AHnb=0](D,BC){} 
 \drawedge[ATnb=1,AHnb=1](D,CA){} 
 \drawedge[ATnb=0,AHnb=1](D,CB){} 
 \drawedge[ATnb=1,AHnb=0](D,CC){}
 \drawedge[ATnb=0,AHnb=0](C,D){}
    \end{picture}
    \vspace{2cm}
    \caption{Fifteen sets associated with an edge $e=(u,v)$.}
\label{fig:edgesets}
\end{center}
\setlength{\unitlength}{1mm}
\end{figure}

We also have sets $A^e_1\gets \seta^u\setminus(\vizinhos(v)\cup\{v\})$, $B^e_1\gets \setb^u\setminus(\vizinhos(v)\cup\{v\})$ and 
$C^e_1\gets \setc^u\setminus(\vizinhos(v)\cup\{v\})$. Finally, we have $A^e_2\gets \seta^v\setminus(\vizinhos(u)\cup\{u\})$, $B^e_2\gets \setb^v\setminus(\vizinhos(u)\cup\{u\})$ and 
$C^e_2\gets \setc^v\setminus(\vizinhos(u)\cup\{u\})$. 

The sets in $\sett^e$ make a partition of
$e$ adjacency.  For instance, a vertex in $v_1\in B_1^e$  belongs to an outside edge from $u$,
a vertex $v_2\in C_2^e$ belongs to an inside edge to $v$. A vertex in $v_3\in AB^e$ belongs to a bidirected edge 
to $u$ and an outside edge from $v$.

Given a set $T\in \sett^e$, we define $n^e_{T}$ as $|T|$.
Given two sets $T_i,T_j\in \sett$, we define $\mv{e}{T_i}{T_j}$ as the number of directed edges from $T_i$ to $T_j$ and 
$\mlv{e}{T_i}{T_j}$ as the number of bidirected edges between $T_i$ and $T_j$. In other words, for all $T_i,T_j\in \sett^e$,
$\mv{e}{T_i}{T_j}\gets |\outset(T_i,T_j)|$ and $\mlv{e}{T_i}{T_j}\gets |\inoutset(T_i,T_j)|$. Thus, if $T_i=AA^e$ and $T_j=AA^e$ then $\mv{e}{AA}{AA}$ is the number of directed edges inside $AA^e$. If $T_i=A_1^e$ and $T_j=A_2^e$, then $\mlv{e}{A_1}{A_2}$ is the number of bidirected edges between $A_1^e$ and $B_2^e$.

Preprocessing these variables is the core technique used to accelerate our algorithm. The variables are processed only once, then they are used to infer the occurrence of motifs. 





Consider an edge $e=(u,v)$ and its neighbor sets; for each $e\in E$, the algorithm will analyze and count the $\ept$ (see Definition \ref{ept}). The patterns containing edge $e$ and a vertex not linked to $e$ are ignored by the $\ept$ set.
Fortunately, all patterns are considered at least by one of its edges, as discussed in Figure~\ref{fig:correction}. Patterns
that are considered in more than one $\ept$ must to be corrected at the end of the algorithm as in the undirected case.

Consider a simple graph $G'$ as $G(V,E)$ induced in $\{u,v\}\cup\vizinhos(u)\cup\vizinhos(v)$. Consider no edges between $\vizinhos(u)$ and $\vizinhos(v)$. This graph is similar to Figure~\ref{fig:edgesets}.  
Note that the set $\ept$ contains vertices $\{u,v\}$ plus two vertices $\{v_1,v_2\}$ in $\vizinhos(u)\cup\vizinhos(v)$. 

Assume that $(u,v)$ is bidirected. To discover the pattern associated
with $\{u,v,v_1,v_2\}$, it is sufficient to know which sets in $\sett^e$ are associated with $v_1$ and $v_2$. For instance, if $v_1\in A_1$ and $v_2\in A_1$, the associated pattern is $S_3$.  If $v_1\in AA^e$ and $v_2\in AA^e$, the associated pattern is $K_4\setminus\{e\}$.  Let $\pt{T_i}{T_j}$ for all $T_i,T_j\in \sett^e$ be the pattern related to $\{u,v,v_1,v_2\}$,
where $e=(u,v)$ and $v_1\in T_i$ and $v_2\in T_j$. The  $\pt{T_i}{T_j}$ for all $T_i$ and $T_j$ is
shown in Table~\ref{table:P}.

%

The algorithm requires the following fact:

\begin{fact}
\label{fact:base}
Let $G'$ be any graph containing a bidirected edge $e=(u,v)$  more  vertices in $(u,v)$ adjacency.  Assume there are no edges in
$\delta(u)\cup\delta(v)$. If it is considered a pattern $\{u,v,v_1,v_2\}$, where $v_1,v_2$ belong to the same set 
$T\in \sett^e $, there are  ${\size{T}}\choose{2}$ occurrences of $\pt{T}{T}$. If $v_1\in T_i$ and $v_2\in T_j$ 
for distinct $T_i,T_j\in \sett^e$, there are
 $\size{T_i}\size{T_j}$ occurrences of $\pt{T_i}{T_j}$. More formally, the frequency
of pattern $P$, $freq(P)$, containing $\{u,v\}$ in $G'$ is
$$freq(P)=\sum_{T\in \sett^e:\pt{T}{T}=P}{{\size{T}}\choose{2}}+\sum_{T_i,T_j\in \sett^e:i<j,P=\pt{T_i}{T_j}}\size{T_i}\size{T_j}.$$
 \end{fact}

It is necessary to define variations of matrix $\pt{T_i}{T_j}$.  If a directed edge $(v_1,v_2)$ is added in $(u,v)$ adjacency, where $v_1\in T_i$ and $v_2\in T_j$, one pattern is removed and one pattern is created. The removed pattern is  defined as $\pt{T_i}{T_j}$ and the created pattern 
is defined as $\pvai{T_i}{T_j}$. If edge $(v_1,v_2)$ is bidirected, the created pattern is defined as $\pvaivolta{T_i}{T_j}$. If $v_1\in T_j$ and $v_2\in T_i$, the created pattern is $\pvolta{T_i}{T_j}$. Figure~\ref{fig:variations}
shows the patterns created when an edge is added between $T_i=A_1^e$ and $T_j=AA$. There is a straightforward
generalization to other possibilities of $T_i$ and $T_j$.

\begin{figure}
\centerline{
\setlength{\unitlength}{3mm}
\begin{tabular}{ccccc} 
\hspace{3cm}  &  \hspace{3cm}  & \hspace{3cm} \\ 
Table~\ref{table:P} ($A_1^e,AA^e)$ & $\pt{A_1^e}{AA^e}$ & $\pvai{A_1^e}{AA^e}$ & $\pvolta{A_1^e}{AA^e}$ & $\pvaivolta{A_1^e}{AA^e}$\\
~\vspace{.6cm}\\
\begin{picture}(0,0)(6,6)\node[Nh=1.6,Nw=1.6,Nmr=2](0)(4,5){$u$}\node[Nh=1.6,Nw=1.6,Nmr=2](1)(8,5){$v$}\node[Nh=1.6,Nw=1.6,Nmr=2](2)(4,9){$v_1$}\node[Nh=1.6,Nw=1.6,Nmr=2](3)(8,9){$v_2$}\drawedge[ATnb=1,AHnb=1,linewidth=.1,AHLength=.7,ATLength=.7](0,2){}\drawedge[ATnb=1,AHLength=0,ATLength=0, linewidth=.1,linecolor=blue](2,3){}\drawedge[ATnb=1,AHnb=1,linewidth=.1,AHLength=.7,ATLength=.7](3,0){}\drawedge[ATnb=1,AHnb=1,linewidth=.1,AHLength=.7,ATLength=.7](1,0){} \drawedge[ATnb=1,AHnb=1,linewidth=.1,AHLength=.7,ATLength=.7](1,3){}
\end{picture} &
\begin{picture}(0,0)(6,6)\node[Nh=1.6,Nw=1.6,Nmr=2](0)(4,5){$u$}\node[Nh=1.6,Nw=1.6,Nmr=2](1)(8,5){$v$}\node[Nh=1.6,Nw=1.6,Nmr=2](2)(4,9){$v_1$}\node[Nh=1.6,Nw=1.6,Nmr=2](3)(8,9){$v_2$}\drawedge[ATnb=1,AHnb=1,linewidth=.1,AHLength=.7,ATLength=.7](0,2){}\drawedge[ATnb=1,AHnb=1,linewidth=.1,AHLength=.7,ATLength=.7](3,0){}\drawedge[ATnb=1,AHnb=1,linewidth=.1,AHLength=.7,ATLength=.7](1,0){}
\drawedge[ATnb=1,AHnb=1,linewidth=.1,AHLength=.7,ATLength=.7](1,3){}\end{picture}  &
\begin{picture}(0,0)(6,6)\node[Nh=1.6,Nw=1.6,Nmr=2](0)(4,5){$u$}\node[Nh=1.6,Nw=1.6,Nmr=2](1)(8,5){$v$}\node[Nh=1.6,Nw=1.6,Nmr=2](2)(4,9){$v_1$}\node[Nh=1.6,Nw=1.6,Nmr=2](3)(8,9){$v_2$}\drawedge[ATnb=1,AHnb=1,linewidth=.1,AHLength=.7,ATLength=.7](0,2){}\drawedge[ATnb=1,AHLength=.7,ATLength=0, linewidth=.1](2,3){}\drawedge[ATnb=1,AHnb=1,linewidth=.1,AHLength=.7,ATLength=.7](3,0){}\drawedge[ATnb=1,AHnb=1,linewidth=.1,AHLength=.7,ATLength=.7](1,0){}
\drawedge[ATnb=1,AHnb=1,linewidth=.1,AHLength=.7,ATLength=.7](1,3){}\end{picture}  &
\begin{picture}(0,0)(6,6)\node[Nh=1.6,Nw=1.6,Nmr=2](0)(4,5){$u$}\node[Nh=1.6,Nw=1.6,Nmr=2](1)(8,5){$v$}\node[Nh=1.6,Nw=1.6,Nmr=2](2)(4,9){$v_1$}\node[Nh=1.6,Nw=1.6,Nmr=2](3)(8,9){$v_2$}\drawedge[ATnb=1,AHnb=1,linewidth=.1,AHLength=.7,ATLength=.7](0,2){}\drawedge[ATnb=1,AHnb=1,linewidth=.1,AHLength=0,ATLength=.7](2,3){}\drawedge[ATnb=1,AHnb=1,linewidth=.1,AHLength=.7,ATLength=.7](3,0){}\drawedge[ATnb=1,AHnb=1,linewidth=.1,AHLength=.7,ATLength=.7](1,0){}
\drawedge[ATnb=1,AHnb=1,linewidth=.1,AHLength=.7,ATLength=.7](1,3){}\end{picture} &
\begin{picture}(0,0)(6,6)\node[Nh=1.6,Nw=1.6,Nmr=2](0)(4,5){$u$}\node[Nh=1.6,Nw=1.6,Nmr=2](1)(8,5){$v$}\node[Nh=1.6,Nw=1.6,Nmr=2](2)(4,9){$v_1$}\node[Nh=1.6,Nw=1.6,Nmr=2](3)(8,9){$v_2$}\drawedge[ATnb=1,AHnb=1,linewidth=.1,AHLength=.7,ATLength=.7](0,2){}\drawedge[ATnb=1,AHnb=1,linewidth=.1,AHLength=.7,ATLength=.7](2,3){}\drawedge[ATnb=1,AHnb=1,linewidth=.1,AHLength=.7,ATLength=.7](3,0){}\drawedge[ATnb=1,AHnb=1,linewidth=.1,AHLength=.7,ATLength=.7](1,0){}
\drawedge[ATnb=1,AHnb=1,linewidth=.1,AHLength=.7,ATLength=.7](1,3){}\end{picture} \\
~
\end{tabular}
}
\caption{Variations of matrix $\pt{T_i}{T_j}$ for $T_i=A_1^e$ and $T_j=AA^e$. }
\label{fig:variations}
\end{figure}
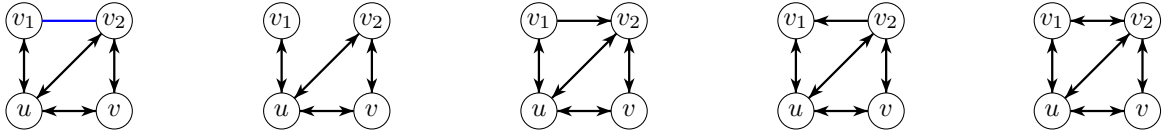

The following theorem is used by the algorithm.

\begin{theorem}\label{lema4d}
Let $G(V,E)$ be a general directed graph and $e=(u,v)$ a bidirected edge in $E$. The patterns occurrences in set $\ept$
is given by the following sum:\\
Start all frequency patterns as zero.\\
\ForEach{$T\in \sett^e$}{
 Increase  $\pt{T}{T}$ occurrence by ${\size{T}\choose 2}-\m{T}{T}-\ml{T}{T}$\\
 Increase  $\pvaivolta{T}{T}$ occurrence by $\ml{T}{T}$\\
 Increase  $\pvai{T}{T}$ occurrence by $\m{T}{T}$
}
\ForEach{\rm $T_i,T_j\in \sett^e, i<j$}{
  Increase   $\pt{T_i}{T_j}$ occurrence by ${\size{T_i}\size{T_j}}-\ml{T_i}{T_j}-\m{T_i}{T_j}-\m{T_j}{T_i}$\\
  Increase  $\pvaivolta{T_i}{T_j}$ occurrence by $\ml{T_i}{T_j}$\\
  Increase  $\pvai{T_i}{T_j}$ occurrence by $\m{T_i}{T_j}$\\
  Increase	  $\pvolta{T_i}{T_j}$ occurrence by $\m{T_j}{T_i}$
}
\end{theorem}
\begin{proof}
This theorem is also proved by induction in the number of edges. 
Let $G'$ be the graph $G$ induced by $\{u,v\}\cup adj(\{u,v\}).$ Suppose that $(u,v)$ is bidirected.
The basic case is if $G'$ does not contain edge in $adj(\{u,v\})$. In this case, the $\ept$  frequencies are given, by construction, by Fact~\ref{fact:base}. The proposed sum is equal to Fact~\ref{fact:base} if all $m^e$ variables are zero, which is the case for the graph at issue.

If one directed edge $(v_1,v_2)$ is added into $T\in \sett^e$, one occurrence of $\pt{T}{T}$ is removed
and one occurrence of $\pvai{T}{T}$ is added. If one bidirected edge $(v_1,v_2)$ is added into $T\in \sett^e$, one occurrence of $\pt{T}{T}$ is removed
and one occurrence of $\pvaivolta{T}{T}$ is added.

If one directed edge $(v_1,v_2)$ is added into two distinct sets $T_i,T_j\in \sett^e$, one occurrence of $\pt{T_i}{T_j}$ is removed
and one occurrence of $\pvai{T_i}{T_j}$ is added. If the added edge is $(v_2,v_1)$, the incremented pattern occurrence
is $\pvolta{T_i}{T_j}$.  If $(v_1,v_2)$ is bidirected, the incremented occurrence is $\pvaivolta{T_i}{T_j}$.
\end{proof}

If the edge $e=(u,v)$ is directed, the pattern associated with $\{u,v,v_1,v_2\}$ must replace the bidirected edge
$(u,v)$ by  a directed one. In this case, the new patterns for $v_1\in T_i$ and $v_2\in T_j$ are represented by $\ptd{T_i}{T_j}$, $\pvaid{T_i}{T_j}$, $\pvoltad{T_i}{T_j}$, $\pvaivoltad{T_i}{T_j}$ instead of $\pt{T_i}{T_j}$, $\pvai{T_i}{T_j}$, $\pvolta{T_i}{T_j}$, $\pvaivolta{T_i}{T_j}$. The results are the same for a directed edge $(u,v)$.

\begin{corollary} \label{coro:a}
Let $G(V,E)$ be a general directed graph and $e=(u,v)$ a {\bf directed} edge in $E$. The patterns occurrences in set $\ept$ can be calculated analogous to Theorem~\ref{lema4d}, but using $\ptd{T_i}{T_j}$, $\pvaid{T_i}{T_j}$, $\pvoltad{T_i}{T_j}$, $\pvaivoltad{T_i}{T_j}$ instead of $\pt{T_i}{T_j}$, $\pvai{T_i}{T_j}$, $\pvolta{T_i}{T_j}$, $\pvaivolta{T_i}{T_j}$ for any $T_i,T_j\in\sett^e$.
\end{corollary}

Algorithm \ref{alg:mine} is used  to count patterns of size 4 by summing the $\ept$ for all $e\in E$ and make a correction if the same induced  has been  was considered many times.

\begin{algorithm}[htb]
\LinesNumbered
\SetAlgoLined
\KwIn{Directed graph $G(V,E)$}
\KwOut{ Histogram to 199 isomorphic patterns to motifs of size 4}
Create a histogram to count isomorphic patterns\\
\ForEach{$e\in E$}{
Calculate the sets $X^e$ for all $X^e\in \sett^e$.\\
Calculate the variables $\m{x}{y}$ and $\ml{x}{y}$ for all $X^e,Y^e\in \sett^e$.\\
 Calculate the frequency involving $e$ and two neighbors using Lemma~\ref{lema4d}\\
 or Corollary~\ref{coro:a}.\\
For each pattern, add this frequency counter to histogram.
}
{\bf if  \rm  the the pattern is related to:}\\
 \Indp {$x\leftrightarrow C_3$ or $S_3$: Divide the frequency counter by 3}\\
  {$C_4$:  Divide the frequency counter by 4}\\
  {$K_4\setminus\{e\}$: Divide the frequency counter by 5}\\
  {$K_4$: Divide the frequency counter by 6}\\
 \Indm   
 \Return{\rm The histogram.}
\caption{ Count 4 Sized Patterns Algorithm.}
\label{alg:mine}
\end{algorithm}

As in the undirected case, no isomorphism-detecting processing is used. The resultant histogram is represented by a vector of integers $h$ with 199 position, one for each distinct
isomorphic pattern of size 4. It is necessary  to preprocess matrices $\pt{}{}$, $\pvai{}{}$, etc., associating each pattern with an arbitrary position in $h$. For instance, it is possible to set the
$\pt{A1}{A2}$ to $h[0]$. The $\pt{A_1}{A_1}$ and $\pt{A_2}{A_2}$ are isomorphic so they can both be associated with $h[1]$.  As a final result, each pattern in the matrices must be \emph{hard coded} in association with a position in the vector $h$ of the histogram, which will be the program output.

The complexity of the algorithm is dominated  by lines 3 and 4, since all other lines are $O(m)$.
We argue that an algorithm similar to  Algorithm~\ref{alg:vars}
in Section~\ref{sec:4u}, can calculate the needed variables in $O(m^2)$.  Thus, it is possible to conclude that the 
proposed algorithm is an $O(m^2)$ algorithm to calculate motifs of size 4 in a directed graph.

\subsection{Counting 5-sized patterns in \emph{undirected} graphs}
\label{subsec:undirected5}


In this section the strategy is extended to motifs of size 5. The following definition is required:

\begin{definition}[$\ppt$]
\label{ppt}
Let $P_3$ be the path graph with three vertices  $p_1,p_2,p_3\in V$.
Given an undirected graph $G(V,E)$, we define a $\ppt$ for any $P_3$ induced in $G(V,E)$ as a set of patterns with five vertices, $\{p_1,p_2, p_3,v_1,v_2\}$, where   $v_1$ and $v_2$  are vertices in $\adj$. The $\ppt$  set has patterns with
an induced  $P_3$  plus two vertices in its adjacency. 
\end{definition}

\begin{definition}[$\kpt$]
\label{kpt}
Let $K_3$ be the clique with vertices $p_1,p_2,p_3\in V$.
Given an undirected graph $G(V,E)$, we define a $\kpt$ for any $K_3$ induced in $G(V,E)$ as a set of patterns with five vertices, $\{p_1,p_2, p_3,v_1,v_2\}$, where   $v_1$ and $v_2$  are in $\adj$. The $\kpt$  set has patterns with
an induced  $K_3$  plus two vertices in its adjacency.  
\end{definition}

The approach to count patterns with five vertices is to count $\ppt$ and $\kpt$ patterns to all $P_3$ and $K_3$ induced 
in $G(V,E)$.  The time to compute the patterns frequency in $\ppt$ and $\kpt$ sets is $O(1)$, considering some preprocessed variables.

The adjacency of $K_3$ and $P_3$ will be partitioned in seven groups.  See Figure \ref{figure:starfive}.
 Let $K_3$ and $P_3$ be composed 
by $(p_1,p_2,p_3)$ and assume that $p_2$ is the central vertex in the $P_3$ case.

Let us define $Z=\vizinhos(p_1)\cap\vizinhos(p_2)\cap\vizinhos(p_3)\setminus\{p_1,p_2,p_3\}$ as the vertices adjacent to $p_1$, $p_2$ and $p_3$,
$X_1=\vizinhos(p_1)\setminus (\vizinhos(p_2)\cup\vizinhos(p_3))$ the vertices only in  $p_1$ adjacency,
$X_2=\vizinhos(p_2)\setminus (\vizinhos(p_1)\cup\vizinhos(p_3))$ the vertices only in  $p_2$ adjacency, and
$X_3=\vizinhos(p_3)\setminus (\vizinhos(p_1)\cup\vizinhos(p_2))$ the vertices only in  $p_3$ adjacency.
Let  $Y_{12}=\vizinhos(p_1)\cap\vizinhos(p_2)\setminus (Z\cup p_3)$ be the vertices only in  $p_1$ and $p_2$ adjacency,
$Y_{13}=\vizinhos(p_1)\cap\vizinhos(p_3)\setminus (Z\cup p_2)$ be the vertices only in  $p_1$ and $p_3$ adjacency, and
$Y_{23}=\vizinhos(p_2)\cap\vizinhos(p_3)\setminus (Z\cup p_1)$  be the vertices only in  $p_2$ and $p_3$ adjacency.
Let $\mathcal{Q}=\{X_1,X_2,X_3,Y_{12},Y_{13},Y_{23},Z\}$.
We define  $\size{x_1}$, $\size{x_2}$ and $\size{x_3}$ as the sizes $|X_1|$, $|X_2|$ and $|X_3|$, respectively, 
  $\size{y_{12}}$, $\size{y_{13}}$ and $\size{y_{23}}$ as the sizes $|Y_{12}|$, $|Y_{13}|$ and $|Y_{23}|$, respectively,
 and $\size{z}$ as the size of $|Z|$.  Let $\mv{}{S_1}{S_2}$ be the number of edges between sets $S_1$ and $S_2$ for all
 $S_1,S_2\in \mathcal{Q}$. 
 
  Table \ref{table:starfive} contains the $\kpt$ frequencies for the graph in Figure~\ref{figure:starfive}.

\begin{figure}[htb]
\centerline{
\setlength{\unitlength}{.7mm}
\begin{picture}(80,45)(0,10)
\node[Nh=3,Nw=3,Nmr=15](B1)(5,-10){} 
\node[Nh=3,Nw=3,Nmr=15](B2)(5,-5){} 
\node[Nh=3,Nw=3,Nmr=15](B3)(5,0){} 
\node[Nh=3,Nw=3,Nmr=15](B4)(5,5){} 
\node[Nh=3,Nw=3,Nmr=15](C3)(70,-10){} 
\node[Nh=3,Nw=3,Nmr=15](C4)(70,-5){} 
\node[Nh=3,Nw=3,Nmr=15](C5)(70,0){} 
\node[Nh=3,Nw=3,Nmr=15](C6)(70,5){} 
\node[Nh=3,Nw=3,Nmr=15](D1)(45,40){} 
\node[Nh=3,Nw=3,Nmr=15](D2)(40,40){} 
\node[Nh=3,Nw=3,Nmr=15](D3)(35,40){} 
\node[Nh=3,Nw=3,Nmr=15](D4)(30,40){} 
\node[Nh=7,Nw=7,Nmr=15](P2)(37,25){$p_2$} 
\node[Nh=7,Nw=7,Nmr=15](P1)(22,7){$p_1$} 
\node[Nh=7,Nw=7,Nmr=15](P3)(52,7){$p_3$} 
\node[Nh=10,Nw=15,Nmr=15](Y23)(62,22){$Y_{23}$} 
\node[Nh=10,Nw=15,Nmr=15](Y12)(12,22){$Y_{12}$} 
\node[Nh=10,Nw=15,Nmr=15](Y13)(37,-5){$Y_{13}$} 
\node[Nh=7,Nw=10,Nmr=15](Z)(37,14){$Z$} 
 \drawedge[ATnb=0,AHnb=0](B1,P1){} 
 \drawedge[ATnb=0,AHnb=0](B2,P1){} 
 \drawedge[ATnb=0,AHnb=0](B3,P1){} 
 \drawedge[ATnb=0,AHnb=0](P1,B4){} 
 \drawedge[ATnb=0,AHnb=0](P3,C3){} 
 \drawedge[ATnb=0,AHnb=0](P3,C4){} 
 \drawedge[ATnb=0,AHnb=0](P3,C5){} 
 \drawedge[ATnb=0,AHnb=0](P3,C6){} 
 \drawedge[ATnb=0,AHnb=0](P3,P2){} 
  \drawedge[ATnb=0,AHnb=0](P3,P2){} 
  \drawedge[ATnb=0,AHnb=0](D1,P2){} 
  \drawedge[ATnb=0,AHnb=0](D2,P2){} 
  \drawedge[ATnb=0,AHnb=0](D3,P2){} 
  \drawedge[ATnb=0,AHnb=0](D4,P2){}
  \drawedge[ATnb=0,AHnb=0](P1,P2){}
   \drawedge[ATnb=0,AHnb=0](P1,Y12){}
   \drawedge[ATnb=0,AHnb=0](P2,Y12){}
   \drawedge[ATnb=0,AHnb=0](P2,Y23){}
   \drawedge[ATnb=0,AHnb=0](P3,Y23){}
   \drawedge[ATnb=0,AHnb=0](P1,Y13){}
   \drawedge[ATnb=0,AHnb=0](P3,Y13){}   
   \drawedge[ATnb=0,AHnb=0](P1,Z){}   
    \drawedge[ATnb=0,AHnb=0](P2,Z){}   
    \drawedge[ATnb=0,AHnb=0](P3,Z){}           
 \drawedge[ATnb=0,AHnb=0,dash={1.5}0](C,D){}  
    \nodelabel[ExtNL=y,NLangle=180,NLdist=3](B3){$X_1$}
  \nodelabel[ExtNL=y,NLangle=0,NLdist=3](C5){$X_3$}
    \nodelabel[ExtNL=y,NLangle=90,NLdist=5](D3){$X_2$}
    \end{picture}~\vspace{2cm}
}
\caption{Sets associated with $P_3$ or $K_3$. If the dashed line is considered, $\{p_1,p_2,p_3\}$ is a $K_3$, otherwise, it is a $P_3$.}
\label{figure:starfive}
\setlength{\unitlength}{1mm}
\end{figure}
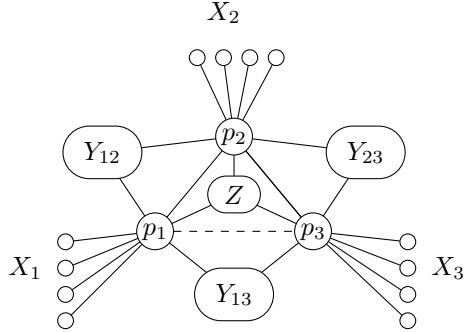

The algorithm to count isomorphic patterns derives from the following theorem.

\begin{theorem}\label{teo5u}
Let $G(V,E)$ be a general undirected graph. The pattern occurrence in set $\kpt$, for any induced clique $(p_1,p_2,p_3)$,
is given by Table~\ref{table:histok}. For any induced $P_3=(p_1,p_2,p_3)$, the pattern occurrence in $\ppt$
is given by Table~\ref{table:histop}. 
%
\end{theorem}

\begin{proof}
This theorem is similar to theorems \ref{teo3} and \ref{teo4u}. It is also proved by induction in the number of edges. 

Consider the $\kpt$ set. To the $\ppt$ set, the proof is the same.
Let $G'$ be the graph $G$ induced by $\{p_1,p_2,p_3\}\cup \adj.$
The basic case is if $G'$ has no edges between vertices in the $\adj$. This graph is  similar to the graph in Figure~\ref{figure:starfive}. In this case, the $\kpt$  frequency are equal to Table~\ref{table:starfive}. Table~\ref{table:histok} corresponds to it if all $m$ variables are zero.

If one edge $e'$ is added into $X_1$, one $\kpt$ pattern  has to be replaced by another. The removed pattern is in Table~\ref{table:histok}, Line 1. The 
pattern added is in Table~\ref{table:histok}, Line 11.
 The same applies to other sets in $\{X_1,X_2,X_3,Y_{12},Y_{13},Y_{23},Z\}$. If one edge is added in $Z$, one
pattern (Table~\ref{table:histok}, Line 10)  has to be replaced by one $K_5$. 
If one edge is added between $X_1$ and $X_3$, 
one pattern    needs to be removed (Table~\ref{table:histok}, Line 2) and other pattern   must be added (Table~\ref{table:histok}, Line 12). 
 Thus, by a straightforward induction in edges in $\adj$, Table \ref{table:histok} presents the   frequency of the set $\kpt$.
 \end{proof}
 

 The algorithm to count isomorphic patterns of size five will list all $K_3$ and $P_3$ induced in $G(V,E)$ and summing up  $\kpt$ and $\ppt$ frequencies for all induced 
 $K_3$ and $P_3$.
 
 %
An induced subgraph pattern can be considered in distinct $\kpt$ and $\ppt$ sets. Thus,
the final histogram needs a small correction. If a pattern appears $a$ times as a  $\kpt$ or a $\ppt$, the final histogram result must be divided by $a$.  For instance, it is shown that 
 $C_5$ is considered in 5 distinct $\ppt$. See Figure \ref{correction:five}.
 The correction analysis needs to be considered for
 all 5-sized isomorphic patterns.

 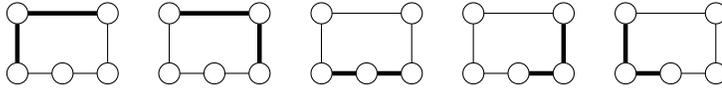
\begin{figure}
\centerline{
\setlength{\unitlength}{.4mm}
\begin{tabular}{cccccc}
\begin{picture}(40,33)(0,0)
\node[Nh=7,Nw=7,Nmr=15](P1)(22,7){} 
\node[Nh=7,Nw=7,Nmr=15](P2)(52,7){} 
\node[Nh=7,Nw=7,Nmr=15](P3)(37,-13){}  
\node[Nh=7,Nw=7,Nmr=15](V1)(52,-13){}  
\node[Nh=7,Nw=7,Nmr=15](V2)(22,-13){}  
  \drawedge[ATnb=0,AHnb=0,linewidth=1.5](P1,P2){}  
   \drawedge[ATnb=0,AHnb=0](P3,V1){}  
   \drawedge[ATnb=0,AHnb=0](P3,V2){}  
   \drawedge[ATnb=0,AHnb=0](P2,V1){}  
      \drawedge[ATnb=0,AHnb=0,linewidth=1.5](P1,V2){}  
   \end{picture}
   &
\begin{picture}(40,33)(0,0)
\node[Nh=7,Nw=7,Nmr=15](P1)(22,7){} 
\node[Nh=7,Nw=7,Nmr=15](P2)(52,7){} 
\node[Nh=7,Nw=7,Nmr=15](P3)(37,-13){}  
\node[Nh=7,Nw=7,Nmr=15](V1)(52,-13){}  
\node[Nh=7,Nw=7,Nmr=15](V2)(22,-13){}  
  \drawedge[ATnb=0,AHnb=0,linewidth=1.5](P1,P2){}  
   \drawedge[ATnb=0,AHnb=0](P3,V1){}  
   \drawedge[ATnb=0,AHnb=0](P3,V2){}  
   \drawedge[ATnb=0,AHnb=0,linewidth=1.5](P2,V1){}  
      \drawedge[ATnb=0,AHnb=0](P1,V2){}  
   \end{picture}
   &
\begin{picture}(40,33)(0,0)
\node[Nh=7,Nw=7,Nmr=15](P1)(22,7){} 
\node[Nh=7,Nw=7,Nmr=15](P2)(52,7){} 
\node[Nh=7,Nw=7,Nmr=15](P3)(37,-13){}  
\node[Nh=7,Nw=7,Nmr=15](V1)(52,-13){}  
\node[Nh=7,Nw=7,Nmr=15](V2)(22,-13){}  
  \drawedge[ATnb=0,AHnb=0,](P1,P2){}  
   \drawedge[ATnb=0,AHnb=0,linewidth=1.5](P3,V1){}  
   \drawedge[ATnb=0,AHnb=0,linewidth=1.5](P3,V2){}  
   \drawedge[ATnb=0,AHnb=0](P2,V1){}  
      \drawedge[ATnb=0,AHnb=0](P1,V2){}  
   \end{picture}
   &
\begin{picture}(40,33)(0,0)
\node[Nh=7,Nw=7,Nmr=15](P1)(22,7){} 
\node[Nh=7,Nw=7,Nmr=15](P2)(52,7){} 
\node[Nh=7,Nw=7,Nmr=15](P3)(37,-13){}  
\node[Nh=7,Nw=7,Nmr=15](V1)(52,-13){}  
\node[Nh=7,Nw=7,Nmr=15](V2)(22,-13){}  
  \drawedge[ATnb=0,AHnb=0,](P1,P2){}  
   \drawedge[ATnb=0,AHnb=0,linewidth=1.5](P3,V1){}  
   \drawedge[ATnb=0,AHnb=0](P3,V2){}  
   \drawedge[ATnb=0,AHnb=0,linewidth=1.5](P2,V1){}  
      \drawedge[ATnb=0,AHnb=0](P1,V2){}  
   \end{picture}
   &
\begin{picture}(40,33)(0,0)
\node[Nh=7,Nw=7,Nmr=15](P1)(22,7){} 
\node[Nh=7,Nw=7,Nmr=15](P2)(52,7){} 
\node[Nh=7,Nw=7,Nmr=15](P3)(37,-13){}  
\node[Nh=7,Nw=7,Nmr=15](V1)(52,-13){}  
\node[Nh=7,Nw=7,Nmr=15](V2)(22,-13){}  
  \drawedge[ATnb=0,AHnb=0,](P1,P2){}  
   \drawedge[ATnb=0,AHnb=0](P3,V1){}  
   \drawedge[ATnb=0,AHnb=0,linewidth=1.5](P3,V2){}  
   \drawedge[ATnb=0,AHnb=0](P2,V1){}  
      \drawedge[ATnb=0,AHnb=0,linewidth=1.5](P1,V2){}  
   \end{picture}\\
   ~\\
    \end{tabular}
}
\caption{The pattern $C_5$ is in $\ppt$  of $P_3$ in bold. It is necessary to divide the final $C_5$ frequency by 5.}
\label{correction:five}
\end{figure}

\begin{algorithm}[htb]
\LinesNumbered
\SetAlgoLined
\KwIn{Undirected graph $G(V,E)$}
\KwOut{ Histogram to 21 isomorphic patterns to motifs of size 5}
Create a histogram to count isomorphic patterns\\
List all induced $K_3=(p_1,p_2,p_3)$ and save in $\mathcal{K}$.\\
List all induced $P_3=(p_1,p_2,p_3)$ and save in $\mathcal{P}$.\\
\ForEach{$t\in \mathcal{K}\cup\mathcal{P}$}{
 Calculate sets $\mathcal{Q}=\{X_1,X_2,X_3,Y_{12},Y_{13},Y_{23},Z\}$.\\
 \ForEach{$S_1,S_2\in \mathcal{Q}$}{
  Calculate the number of edges between $S_1$ and $S_2$, $\mv{}{S_1}{S_2}$
}
 Calculate the frequency involving $t$ and two neighbors using tables~\ref{table:histok} and \ref{table:histop} .\\
For each pattern, add this frequency to histogram.\\
}
\ForEach{pattern in the histogram}{
	Divide the frequency by the constant $a$, correcting the repetition in several $\ppt$ and $\kpt$
}
{\bf  
    \Return{\rm The histogram.}
 }
\caption{   Count 5 Sized Patterns Algorithm.}
\label{alg:und5}
\end{algorithm}

The Algorithm~\ref{alg:und5} counts subgraph patterns  of size 5.  In lines 2 and 3, 
a list of all $P_3$ and $K_3$ induced in $G(V,E)$ is computed.
Algorithm~\ref{alg:calckp} does this task in $O(\sum_{\{u,v\}\in E}|\vizinhos(v)|+|\vizinhos(v)| )$. As $|\vizinhos(u)|+|\vizinhos(v)|=O(n)$, the complexity
of Algorithm~\ref{alg:calckp} is $O(mn)$. Note that the size of   $\mathcal{K}$ and $\mathcal{P}$ are bounded by $\Theta(nm)$.

\begin{algorithm}[htb]
\LinesNumbered
\SetAlgoLined
\KwIn{Undirected graph $G(V,E)$}
\KwOut{ The set $\mathcal{P}$ of all induced $P_3$ and $\mathcal{K}$ of all induced $K_3$ in $G(V,E)$.}
Set $\mathcal{P}\gets \emptyset$, $\mathcal{K}\gets \emptyset$.\\
\ForEach{$e=\{u,v\}\in E$}{
\ForEach{$x\in adj(\{u,v\})$}{
   \If{(x,u,v) is a $P_3$ and (x,u,v) is not considered yet }{
      $\mathcal{P}\gets\mathcal{P}\cup\{(x,u,v)\}$
   }
   \If{(x,u,v) is a $K_3$ and (x,u,v) is not considered yet}{
      $\mathcal{K}\gets\mathcal{K}\cup\{(x,u,v)\}$
   } 
} 
}

{\bf  
    \Return{\rm $\mathcal{K}$ and $\mathcal{P}$}
 }
\caption{   List all induced $P_3$ and $K_3$.}
\label{alg:calckp}
\end{algorithm}

Algorithm~\ref{alg:und5}, Line 5, consists of calculating the sets  $\mathcal{Q}=\{X_1,X_2,X_3,Y_{12},Y_{13},Y_{23},Z\}$ for a single $P_3$ or $K_3$. This cost is dominated
by $|\delta(p_1)|+|\delta(p_2)|+|\delta(p_3)|=O(n)$. It is a simple intersection of sets procedures and can be 
efficiently computed using a hash.


As in the previous section, no isomorphism detecting algorithm is used.
The histogram is represented by a vector $h$ with 21 positions.  The algorithm associates each pattern 
with an arbitrary \emph{hard coded} position in the vector. 

Algorithm~\ref{alg:varsfive} computes the needed variables, according to Line 7 of Algorithm~\ref{alg:und5}. 
The algorithm
has an $O(m)$ complexity for each induced $K_3$ and $P_3$.

\begin{algorithm}[htb]
\LinesNumbered
\SetAlgoLined
\KwIn{Undirected graph $G(V,E)$ and a conneceted $(p_1,p_2,p_3)\in\mathcal{K}\cup\mathcal{P}$}
\KwOut{Variables   $\{\mv{}{x_1}{x_1},\mv{}{x_1}{x_2},\ldots,\mv{}{z}{z}\}$.}
 All variables in $\{\mv{}{x_1}{x_1},\mv{}{x_1}{x_2},\ldots,\mv{}{z}{z}\}$ start with zero.\\
\ForEach{$(x,y)\in E$ }{
\begin{tabular}{lc}
	\hline
	\bf{If vertices~~~~~}  &  \bf{Var to increment} \\\hline
	 $x \in X_1$ and $y\in X_1$ & $\mv{}{x_1}{x_1}$  \\
	 $x \in X_1$ and $y\in X_2$ & $\mv{}{x_1}{x_2}$\\
	 $x \in X_1$ and $y\in X_3$ & $\mv{}{x_1}{x_3}$\\
	 $x \in X_1$ and $y\in Y_{12}$ & $\mv{}{x_1}{y_{12}}$\\
	 \vdots&\vdots\\
	 $x \in Z$ and $y\in Y_{23}$ & $\mv{}{z}{y_{23}}$\\	 	 
	 $x \in Z$ and $y\in Z$ & $\mv{}{z}{z}$\\	 	 
	\hline 
	~\\
\end{tabular}
}
\Return{ \rm The computed variables.}
\caption{ Create $\{\mv{}{x_1}{x_1},\mv{}{x_1}{x_2}\ldots,\mv{}{z}{z}\}$  variables for a given $e\in E$.}
\label{alg:varsfive}
\end{algorithm}

We can conclude that  Algorithm~\ref{alg:und5} counts isomorphic pattern motifs of size 5 in $O(m^2n)$
in an undirected graph $G(V,E)$. Moreover, the additional memory to store the variables is $\Theta(mn)$.

\subsection{Counting 5-sized patterns in \emph{directed} graphs}
\label{subsec:directed5}

Similarly to Section~\ref{subsec:directed4}, no new concept is needed to extend the previous algorithm to the directed version. Both algorithms are very similar. The main difference is that the 
directed case requires dealing with a large number of sets. The neighborhood of $(p_1,p_2,p_3)$ is partitioned in 63 sets. It is necessary 
to calculate how many edges are found between each 2,016 combinations of two such sets. The histogram vector has 9,364 positions, one for each distinct
isomorphic pattern. Regardless of such difficulty, the algorithm deal with sets. The number of sets or set combinations in the neighborhood of $(p_1,p_2,p_3)$ is always $O(1)$. 

The following definitions are required:

\begin{definition}
Let $\Pt$ be a set of all isomorphic patterns of size 3. The set $\Pt$ contains thirteen elements:
{\setlength{\unitlength}{.3mm}
$$\Pt=\left \{\begin{tabular}{cccccccc}
\begin{picture}(40,33)(0,-20)
\node[Nh=12,Nw=12,Nmr=15](P1)(22,7){\footnotesize$p_1$} 
\node[Nh=12,Nw=12,Nmr=15](P2)(52,7){\footnotesize$p_2$} 
\node[Nh=12,Nw=12,Nmr=15](P3)(37,-13){\footnotesize$p_3$}  
 \drawedge[ATnb=1,AHnb=1](P1,P2){}  
 \drawedge[ATnb=1,AHnb=1](P2,P3){}  
  \end{picture}
  & ,
\begin{picture}(40,33)(0,-20)
\node[Nh=12,Nw=12,Nmr=15](P1)(22,7){\footnotesize$p_1$} 
\node[Nh=12,Nw=12,Nmr=15](P2)(52,7){\footnotesize$p_2$} 
\node[Nh=12,Nw=12,Nmr=15](P3)(37,-13){\footnotesize$p_3$}  
\drawedge[ATnb=1,AHnb=0](P1,P2){}  
 \drawedge[ATnb=1,AHnb=1](P2,P3){}  
   \end{picture}
   &,
\begin{picture}(40,33)(0,-20)
\node[Nh=12,Nw=12,Nmr=15](P1)(22,7){\footnotesize$p_1$} 
\node[Nh=12,Nw=12,Nmr=15](P2)(52,7){\footnotesize$p_2$} 
\node[Nh=12,Nw=12,Nmr=15](P3)(37,-13){\footnotesize$p_3$}  
\drawedge[ATnb=0,AHnb=1](P1,P2){}  
 \drawedge[ATnb=1,AHnb=1](P2,P3){}  
   \end{picture}
   &,
\begin{picture}(40,33)(0,-20)
\node[Nh=12,Nw=12,Nmr=15](P1)(22,7){\footnotesize$p_1$} 
\node[Nh=12,Nw=12,Nmr=15](P2)(52,7){\footnotesize$p_2$} 
\node[Nh=12,Nw=12,Nmr=15](P3)(37,-13){\footnotesize$p_3$}  
\drawedge[ATnb=0,AHnb=1](P1,P2){}  
 \drawedge[ATnb=0,AHnb=1](P2,P3){}  
   \end{picture}
   &,
\begin{picture}(40,33)(0,-20)
\node[Nh=12,Nw=12,Nmr=15](P1)(22,7){\footnotesize$p_1$} 
\node[Nh=12,Nw=12,Nmr=15](P2)(52,7){\footnotesize$p_2$} 
\node[Nh=12,Nw=12,Nmr=15](P3)(37,-13){\footnotesize$p_3$}  
\drawedge[ATnb=0,AHnb=1](P1,P2){}  
 \drawedge[ATnb=1,AHnb=0](P2,P3){}  
   \end{picture}
&,
\begin{picture}(40,33)(0,-20)
\node[Nh=12,Nw=12,Nmr=15](P1)(22,7){\footnotesize$p_1$} 
\node[Nh=12,Nw=12,Nmr=15](P2)(52,7){\footnotesize$p_2$} 
\node[Nh=12,Nw=12,Nmr=15](P3)(37,-13){\footnotesize$p_3$}  
\drawedge[ATnb=1,AHnb=0](P1,P2){}  
 \drawedge[ATnb=0,AHnb=1](P2,P3){}  
   \end{picture}\\
\begin{picture}(40,33)(0,-20)
\node[Nh=12,Nw=12,Nmr=15](P1)(22,7){\footnotesize$p_1$} 
\node[Nh=12,Nw=12,Nmr=15](P2)(52,7){\footnotesize$p_2$} 
\node[Nh=12,Nw=12,Nmr=15](P3)(37,-13){\footnotesize$p_3$}  
\drawedge[ATnb=1,AHnb=1](P1,P2){}  
 \drawedge[ATnb=1,AHnb=1](P2,P3){}  
 \drawedge[ATnb=1,AHnb=1](P1,P3){} 
 \end{picture}
 &,
 \begin{picture}(40,33)(0,-20)
\node[Nh=12,Nw=12,Nmr=15](P1)(22,7){\footnotesize$p_1$} 
\node[Nh=12,Nw=12,Nmr=15](P2)(52,7){\footnotesize$p_2$} 
\node[Nh=12,Nw=12,Nmr=15](P3)(37,-13){\footnotesize$p_3$}  
\drawedge[ATnb=1,AHnb=1](P1,P2){}  
 \drawedge[ATnb=1,AHnb=1](P2,P3){}  
 \drawedge[ATnb=1,AHnb=0](P1,P3){} 
 \end{picture}
 &,
  \begin{picture}(40,33)(0,-20)
\node[Nh=12,Nw=12,Nmr=15](P1)(22,7){\footnotesize$p_1$} 
\node[Nh=12,Nw=12,Nmr=15](P2)(52,7){\footnotesize$p_2$} 
\node[Nh=12,Nw=12,Nmr=15](P3)(37,-13){\footnotesize$p_3$}  
\drawedge[ATnb=1,AHnb=1](P1,P2){}  
 \drawedge[ATnb=1,AHnb=0](P2,P3){}  
 \drawedge[ATnb=1,AHnb=0](P1,P3){} 
 \end{picture}
 &,
  \begin{picture}(40,33)(0,-20)
\node[Nh=12,Nw=12,Nmr=15](P1)(22,7){\footnotesize$p_1$} 
\node[Nh=12,Nw=12,Nmr=15](P2)(52,7){\footnotesize$p_2$} 
\node[Nh=12,Nw=12,Nmr=15](P3)(37,-13){\footnotesize$p_3$}  
\drawedge[ATnb=1,AHnb=1](P1,P2){}  
 \drawedge[ATnb=1,AHnb=0](P2,P3){}  
 \drawedge[ATnb=0,AHnb=1](P1,P3){} 
 \end{picture}
 &,
   \begin{picture}(40,33)(0,-20)
\node[Nh=12,Nw=12,Nmr=15](P1)(22,7){\footnotesize$p_1$} 
\node[Nh=12,Nw=12,Nmr=15](P2)(52,7){\footnotesize$p_2$} 
\node[Nh=12,Nw=12,Nmr=15](P3)(37,-13){\footnotesize$p_3$}  
\drawedge[ATnb=1,AHnb=1](P1,P2){}  
 \drawedge[ATnb=0,AHnb=1](P2,P3){}  
 \drawedge[ATnb=0,AHnb=1](P1,P3){} 
 \end{picture}
    \end{tabular}~~~~\right \}
$$}
\end{definition}

\begin{definition}[$\Dpt$]
\label{dpt}
Let $\Delta$ be one element in $\Pt$.
Given a directed graph $G(V,E)$, we define a $\Dpt$ for any $(p_1,p_2,p_3)$ induced in $G(V,E)$ and isomorphic to $\Delta$, as a set of patterns with five vertices, $\{p_1,p_2, p_3,v_1,v_2\}$, where   $v_1$ and $v_2$  are vertices in $\adj$. A $\Dpt$ is a set of patterns with
an induced  $\Delta$  plus two vertices in its adjacency. 
\end{definition}

Considering a pattern $\Delta\in \Pt$, it is possible to define 63 sets associated with it. Consider 
three variables $\alpha$, $\beta$ and $\gamma$  with domain $\{A,B,C,\emptyset\}$.
The variables $\alpha$, $\beta$ and $\gamma$ define the neighborhood of a given vertex with respect to $p_1$, $p_2$ and $p_3$, respectively.
Consider the vertices in $\adj$. The Algorithm~\ref{partfive} partitions such vertices
in 63 sets.

\begin{algorithm}[htb]
\LinesNumbered
\SetAlgoLined
\KwIn{Directed graph $G(V,E)$ and an induced connected subgraph $(p_1,p_2,p_3)$}
\KwOut{The vertices in $\adj$ partitionated in 63 sets.}
\ForEach{
$\alpha\in\{A,B,C,\emptyset\}$,$\beta\in\{A,B,C,\emptyset\}$,$\gamma\in\{A,B,C,\emptyset\}$
}{
$set(\alpha,\beta,\gamma)\gets\emptyset$
}
\ForEach{
$v\in \adj$
}{
If $v\in\left\{
\begin{tabular}{ll}
$\seta(p_1)$, & $\alpha\gets A$\\
$\setb(p_1)$, & $\alpha\gets B$\\
$\setc(p_1)$, & $\alpha\gets C$\\
o.c. , & $\alpha\gets \emptyset$
\end{tabular}
\right.$\\
If $v\in\left\{
\begin{tabular}{ll}
$\seta(p_2)$, & $\beta\gets A$\\
$\setb(p_2)$, & $\beta\gets B$\\
$\setc(p_2)$, & $\beta\gets C$\\
o.c. , & $\beta\gets \emptyset$
\end{tabular}
\right.$\\
If $v\in\left\{
\begin{tabular}{ll}
$\seta(p_3)$, & $\gamma\gets A$\\
$\setb(p_3)$, & $\gamma\gets B$\\
$\setc(p_3)$, & $\gamma\gets C$\\
o.c. , & $\gamma\gets \emptyset$
\end{tabular}
\right.$\\
$set(\alpha,\beta,\gamma)\gets set(\alpha,\beta,\gamma)\cup\{v\}$\\
}
\Return{$set(\alpha,\beta,\gamma)$ for all $\alpha$, $\beta$ and $\gamma$ in $\{A,B,C,\emptyset\}$. The set $set(\emptyset,\emptyset,\emptyset)$ is empty and can be ignored.}

\label{partfive}
\caption{Algorithm to compute sets $\mathcal{Q}$ of a given pattern $(p_1,p_2,p_3)\in \Pt$.}
\end{algorithm}

The role of partition is simple. If $v\in set(A,A,A)$, it is in the set with bidirected edges to $p_1,p_2$ and $p_3$. If $v\in set(\emptyset,\emptyset, C)$, it is not connected
to $p_1$, not connected to $p_2$ and connected to $p_3$ by a directed edge to $p_3$.

Let $\mathcal{Q}=\{set(\alpha,\beta,\gamma)~|~\forall \alpha,\beta,\gamma\in\{A,B,C,\emptyset\}\}$.
Given a set $T\in \mathcal{Q}$, we define $\size{T}$ as $|T|$.
Given two sets $T_i,T_j\in \mathcal{Q}$, we define $\mv{}{T_i}{T_j}$ as the number of directed edges from $T_i$ to $T_j$ and 
$\mlv{}{T_i}{T_j}$ as the number of bidirected edges between $T_i$ and $T_j$. In other words, for all $T_i,T_j\in \mathcal{Q}$,
$\mv{}{T_i}{T_j}\gets |\outset(T_i,T_j)|$ and $\mlv{}{T_i}{T_j}\gets |\inoutset(T_i,T_j)|$. Thus, if $T_i=set(A,A,A)$ and $T_j=set(A,A,A)$ then $\mv{}{set(A,A,A)}{set(A,A,A)}$ is the number of directed edges inside $set(A,A,A)$. If $T_i=set(A,A,A)$ and $T_j=set(\emptyset,\emptyset, C)$ then $\mlv{}{set(A,A,A)}{set(\emptyset,\emptyset, C)}$ is the number of bidirected edges between $set(A,A,A)$ and $set(\emptyset,\emptyset, C)$.

As in Section~\ref{subsec:directed4}, preprocessing these variables is the core technique used to accelerate our algorithm. The variables are processed only once. Next, they are used to calculate the occurrence of motifs.

Consider a pattern $\Delta\in \Pt$ and its neighboring sets; for each $\Delta$ induced in $G(V,E)$, the algorithm will analyze and count the $\Dpt$ (see Definition \ref{dpt}). The  patterns
that are considered in more than one $\Dpt$ must to be corrected at the end of the algorithm, as in the undirected case.

Consider a simple graph $G'$ as $G(V,E)$ induced in $\{p_1,p_2,p_3\}\cup \adj$, where $(p_1,p_2,p_3)$ is isomorphic to $\Delta$. Consider no edges in $\adj$. 

To discover the pattern associated
with $\{p_1,p_2,p_3,v_1,v_2\}$, it is sufficient to known $\Delta$ and which sets in $\mathcal{Q}$ are associated with $v_1$ and $v_2$. 

Let $K_3$ be the complete graph of size 3, $K_5\setminus e$ be the complete graph of size 5 with one arbitrary bidirected edge removed
and $P_n$ be the path graph with $n$ vertices. All edges in $K_3$, $K_5\setminus e$ and $P_n$ are bidirected.
For instance, if $v_1\in set(A,A,A)$ and $v_2\in set(A,A,A)$ and $\Delta=K_3$,
 the associated pattern is $K_5\setminus e$.  If $v_1\in set(A,\emptyset,\emptyset)$ and $v_2\in set(\emptyset,\emptyset, A)$, and $\Delta=P_3$,
 the associated pattern is $P_5$.  Let $\ptaux{\Delta}{T_i}{T_j}$ for all $\Delta\in \Pt$ and $T_i,T_j\in \mathcal{Q}$ be the pattern related to $\{p_1,p_2,p_3,v_1,v_2\}$,
where $(p_1,p_2,p_3)$ is isomorphic to $\Delta$ and $v_1\in T_i$ and $v_2\in T_j$. The  $\ptaux{\Delta}{T_i}{T_j}$ has a simple rule to be created, but
it is a matrix with $13\times 63\times 63$ dimensions, resulting in a matrix with 51,597 precomputed patterns.

The following fact is closely related to Fact \ref{fact:base}.

\begin{fact}
\label{fact:basefive}
Consider a graph induced in $\{p_1,p_2,p_3\}$ isomorphic to $\Delta\in \Pt$.
Let $G'$ be any graph containing an induced isomorphic pattern $\Delta\in \Pt$  plus two vertices in $\Delta$ adjacency.  Assume there are no edges in
$\adj$. If it is considered a pattern $\{p_1,p_2,p_3,v_1,v_2\}$, where $v_1,v_2$ belong to the same set 
$T\in \mathcal{Q} $, there are  ${\size{T}}\choose{2}$ occurrences of $\ptaux{\Delta}{T}{T}$. If $v_1\in T_i$ and $v_2\in T_j$ 
for distinct $T_i,T_j\in \mathcal{Q}$, there are
 $\size{T_i}\size{T_j}$ occurrences of $\ptaux{\Delta}{T_i}{T_j}$. More formally, the frequency
of pattern $P$ contained in $\Dpt$ is
$$freq(P)=\sum_{T\in \mathcal{Q}:\ptaux{\Delta}{T}{T}=P}{{\size{T}}\choose{2}}+\sum_{T_i,T_j\in \mathcal{Q}:i<j,P=\ptaux{\Delta}{T_i}{T_j}}\size{T_i}\size{T_j}.$$
 \end{fact}

Exactly as in Section \ref{subsec:directed4}, it is necessary to define variations of matrix $\ptaux{\Delta}{T_i}{T_j}$.  If a directed edge $(v_1,v_2)$ is added in $\adj$, where $v_1\in T_i$ and $v_2\in T_j$, one  $\ptaux{\Delta}{T_i}{T_j}$ is removed and another pattern is created. The created pattern 
is defined as $\pvaiaux{\Delta}{T_i}{T_j}$. 
If $v_1\in T_j$ and $v_2\in T_i$, the created pattern is $\pvoltaaux{\Delta}{T_i}{T_j}$.
If edge $(v_1,v_2)$ is bidirected, the created pattern is defined as $\pvaivoltaaux{\Delta}{T_i}{T_j}$.  Figure~\ref{fig:variationsfive}
shows the patterns created when $\Delta=P_3$ and an edge is added between $T_i=set(A,B,\emptyset)$ and $T_j=set(\emptyset,\emptyset,C)$. There is a straightforward
generalization to other possibilities of $T_i$ and $T_j$ in $\mathcal{Q}$ and $\Delta\in \Pt$.

The following lemma is used by the algorithm.

\begin{lemma}\label{lema5d}
Let $G(V,E)$ be a general directed graph and $\Delta\in \Pt$ an induced connected graph of size 3. The patterns occurrences in set $\Dpt$
is given by the following sum:\\
Start all frequency patterns as zero.\\
\ForEach{$T\in \mathcal{Q}$}{
 Increase  $\ptaux{\Delta}{T}{T}$ occurrence by ${\size{T}\choose 2}-\m{T}{T}-\ml{T}{T}$\\
 Increase  $\pvaivoltaaux{\Delta}{T}{T}$ occurrence by $\ml{T}{T}$\\
 Increase  $\pvaiaux{\Delta}{T}{T}$ occurrence by $\m{T}{T}$
}
\ForEach{\rm $T_i,T_j\in \mathcal{Q}, i<j$}{
  Increase   $\ptaux{\Delta}{T_i}{T_j}$ occurrence by ${\size{T_i}\size{T_j}}-\ml{T_i}{T_j}-\m{T_i}{T_j}-\m{T_j}{T_i}$\\
  Increase  $\pvaivoltaaux{\Delta}{T_i}{T_j}$ occurrence by $\ml{T_i}{T_j}$\\
  Increase  $\pvaiaux{\Delta}{T_i}{T_j}$ occurrence by $\m{T_i}{T_j}$\\
  Increase	  $\pvoltaaux{\Delta}{T_i}{T_j}$ occurrence by $\m{T_j}{T_i}$
}
\end{lemma}
\begin{proof}
Similar to Theorem \ref{lema4d}.
\end{proof}

The Algorithm \ref{alg:fived}  counts patterns of size 5 by summing the $\Dpt$ for all $\Delta\in \Pt$ induced in $G(V,E)$. 
The algorithm applies a correction if the same induced subgraph has been considered many times.
Note that the sum in Lemma \ref{lema5d} takes $O(1)$, assuming precomputed variables.

\begin{algorithm}[htb]
\LinesNumbered
\SetAlgoLined
\KwIn{Directed graph $G(V,E)$}
\KwOut{ Histogram to 9364 isomorphic patterns to motifs of size 5}
Create a histogram to count isomorphic patterns\\
Create an empty list $\mathcal{L}_{\Delta}$ for all $\Delta\in \Pt$\\
List all connected $(p_1,p_2,p_3)$ induced in $G(V,E)$ and add to the 
correspondent $\mathcal{L}_{\Delta}$\\
\ForEach{$\Delta\in \Pt$}{
\ForEach{$t\in \mathcal{L}_{\Delta}$}{
 Calculate sets $\mathcal{Q}=\{set(\alpha,\beta,\gamma)~|~\forall \alpha,\beta\gamma\in\{A,B,C,\emptyset\}\}$.\\
 \ForEach{$S_1,S_2\in \mathcal{Q}$}{
  Calculate the number of edges between $S_1$ and $S_2$, $\mv{}{S_1}{S_2}$
}
Calculate the frequency in $\Dpt$ using Lemma~\ref{lema5d}
}
}
\ForEach{pattern in the histogram}{
	Divide the frequency by the constant $a$, correcting the repetition in several $\Dpt$
}
{\bf  
    \Return{\rm The histogram.}
 }
\caption{ Count 5 Sized Patterns Algorithm.}
\label{alg:fived}
\end{algorithm}

As in the previous cases, there is no isomorphism-detecting algorithm. It is necessary to preprocess matrices $\ptaux{}{}$, $\pvaiaux{}{}$ and $\pvoltaaux{}$, associating each pattern with an arbitrary position in the histogram. 
Such matrices are computed once. Every execution of $\acc$ uses the same matrices.

The complexity of the algorithm is dominated  by lines 6 and 8.

Algorithm~\ref{partfive}  computes the sets $\mathcal{Q}$ for a given $(p_1,p_3,p_3)$ in $O(|\delta(p_1)|+|\delta(p_2)|+|\delta(p_3)|)=O(n)$.
Algorithm~\ref{alg:varsfive}  can be also adapted to compute the required variables in 
$O(m)$ for a given $(p_1,p_2,p_3)$.
Since all other lines are $O(1)$ for $\Dpt$, we conclude that the 
proposed algorithm is an $O(m^2n)$ algorithm to calculate  5-sized motifs in directed graphs.

\color{black}
\section{Results and Discussion}
\label{sec:exp}

This section compares the computational results of our algorithm, that we called {\bf $\acc$} (accelerated Motif), with FANMOD~\cite{wernicke2006fanmod} and
Kavosh~\cite{kashani2009kavosh}. The tools FANMOD and Kavosh were chosen because they are two of the fastest available motif finder programs~\cite{Elisabeth2011}.

The instances were arbitrarily selected from a wide range of motif applications. They are selected 
from  open complex network databases such as Pajek and  Uri Alon 
datasets~\cite{AlonDataset,Pajek2006}.  We preprocessed the instances, removing vertices with zero neighbors.

The implemented algorithms are devised to motifs of sizes $3$, $4$ and $5$.
To ensure replicability and better evaluation, we have provided the input tested graphs and the java byte-code of implemented algorithms available at \url{http://www.ft.unicamp.br/~meira/accmotifs}. 

All the tests were performed in an 
Intel(R) Core(TM)2 Quad CPU Q8200, 2.33 GHz, 2 GB RAM, using an algorithm implemented in Java language.
We set FANMOD and Kavosh with the full enumeration parameter. Thus,  FANMOD, Kavosh  and $\acc$ solve the 
same problem, which consists in counting \emph{all} subgraph of the selected size. The final histogram is exactly the same for a given graph, ensuring the correctness of $\acc$.


Tables~\ref{temposk3},~\ref{temposk4} and ~\ref{temposk5} show the execution time of FANMOD, Kavosh and $\acc$ for $k$ equal to 3, 4 and 5, respectively.
The algorithm is executed in the original graph and in a set of random graphs. The number of random 
graphs is 100, 10 and 5 for $k$ equal to $3$, $4$ and $5$, respectively.
 The  time reported is the average considering the original and the random graphs.  
In this experiment, only the execution time is considered to enumerate all subgraphs. The time to generate the random graphs is not considered.


Each round consists in the subgraph enumeration in the original and in the random graphs.
We repeated the execution by five rounds per instance. Tables~\ref{temposk3},~\ref{temposk4} and ~\ref{temposk5} contain 
the average and the deviation factor of these five measurements.

We limited the CPU time to 5 hours per graph for the sake of convenience.
In tables~\ref{temposk3},~\ref{temposk4} and ~\ref{temposk5}, it is possible to observe that the proposed algorithms were expressively faster than FANMOD and Kavosh in almost all tested instances.
For $k=3$ and instance Airport	 \cite{Tore}, $\acc$ spent $35\pm0.2$ms/graph, Kavosh spent $1,250\pm6$ms/graph and FANMOD spent 
 $5,772\pm2$ms/graph.
For $k=4$ and instance ODLIS \cite{Pajek2006}, $\acc$ spent $2,605\pm115$ms/graph, Kavosh spent  $ 	 210,015	 \pm 	 837$ms/graph and FANMOD spent 
 $630,936\pm2,914$ms/graph.
 For $k=5$ and instance  California \cite{Pajek2006}, $\acc$ spent $76\pm0.4$ s/graph, Kavosh spent  $1,532\pm 12$s/graph and FANMOD spent 
 $2,376\pm12$s/graph. The performance gain of $\acc$ was consistently observed in all tested instances.


\section{Conclusion}

\label{sec:conc}
Three new exact algorithms were presented to calculate motifs using combinatorial techniques.
The algorithms
have complexity $O(m\sqrt{m})$ to count isomorphic patterns of size 3 and $O(m^2)$ to count isomorphic patterns of 
size~4 and $O(m^2n)$ to count isomorphic patterns of size 5. Computational results show that the proposed exact algorithms are expressively faster than the known techniques (e.g., FANMOD or Kavosh).


\section{Availability and requirements}

Project name: $\acc$ - Accelerated Motif Detection Using Combinatorial Techniques \\
Project home page:  \url{http://www.ft.unicamp.br/~meira/accmotifs}. \\
Operating system(s):  Platform independent\\
Programming language:  Java\\
Other requirements: e.g. Java 1.6. 1GB Ram. \\
License: Freeware. \\ 
Any restrictions to use by non-academics: None.

\section{Authors' contributions}

All authors are involved in algorithms discussion and manuscript writing. L. A. A. Meira and V.  R. M\'{a}ximo are involved in the java implementation.


\section{Acknowledgements}

The authors would like to thank Professor Paulo Bandiera Paiva by initial problem statement and data acquisition. This research was partially supported by CNPq-Brazil.

\section{Competing interests}

The authors declare that they have no competing interests.

\bibliographystyle{plain}

\bibliography{principal}


\begin{table}[h!]
\begin{center}
\begin{tabular}{c|c|c}
$ \seta^v=\inoutset(v)$ & $\setb^v=\outset(v)$ & $\setc^v=\inset(v)$ \\\hline
$n_a^v = |\seta^v|$&$n_b^v=|\setb^v|$&$n_c^v=|\setc^v|$\\\hline
$\mv{v}{a}{b}=|\outset( \seta^v, \setb^v)|~~$ & $~~\mv{v}{a}{c}=|\outset( \seta^v, \setc^v)|~~$& $~~\mv{v}{b}{c}=|\outset( \setb^v, \setc^v)|$\\\hline
$\mv{v}{b}{a}=|\outset( \setb^v, \seta^v)|~~$ & $~~\mv{v}{c}{a}=|\outset( \setc^v, \seta^v)|~~$& $~\mv{v}{c}{b}=|\outset( \setc^v, \setb^v)|$\\\hline
$\mv{v}{a}{a}=|\outset( \seta^v, \seta^v)|~~$ & $~~\mv{v}{b}{b}=|\outset( \setb^v, \setb^v)|~~$& $~~\mv{v}{c}{c}=|\outset( \setc^v, \setc^v)|$\\\hline
$\mlv{v}{a}{b}=\mlv{v}{b}{a}=|\inoutset( \seta^v, \setb^v)|~~$ & $~~\mlv{v}{a}{c}=\mlv{v}{c}{a}=|\inoutset( \seta^v, \setc^v)|~~$& $~~\mlv{v}{b}{c}=\mlv{v}{c}{b}=|\inoutset( \setb^v, \setc^v)|$\\\hline
$\mlv{v}{a}{a}=|\inoutset( \seta^v, \seta^v)|~~$ & $~~\mlv{v}{b}{b}=|\inoutset( \setb^v, \setb^v)|~~$& $~~\mlv{v}{c}{c}=|\inoutset( \setc^v, \setc^v)|$\\
\end{tabular}
\caption{Variables to vertex $v$.}
\label{table:vars}
\end{center}
\end{table}

\begin{table}[h!]
\begin{center}
\begin{tabular}{cc}
\begin{tabular}{lcc}
	\hline
	{\bf Pattern~~~~~}  & {\bf Frequency}  &  {\bf Line}\\
	\hline
	$\vspace{.2cm} o\leftrightarrow o \rightarrow o$  & $\displaystyle n_a^v n_b^v-\mv{v}{a}{b}-\mv{v}{b}{a}-\mlv{v}{a}{b}$&1\\
	$\vspace{.2cm} o\leftrightarrow o \leftarrow o$  & $\displaystyle n_a^v n_c^v-\mv{v}{a}{c}-\mv{v}{c}{a}-\mlv{v}{a}{c}$&2\\
	$\vspace{.2cm}o\rightarrow o \rightarrow o$  & $\displaystyle n_b^v n_c^v-\mv{v}{b}{c}-\mv{v}{c}{b}-\mlv{v}{b}{c}$&3\\
	$ o\leftrightarrow o \leftrightarrow o$  & $\displaystyle{{ n_a^v}\choose{2}}-\mv{v}{a}{a}-\mlv{v}{a}{a}$&4\\
	$o\leftarrow o \rightarrow o$  & $\displaystyle{{ n_b^v}\choose{2}}-\mv{v}{b}{b}-\mlv{v}{b}{b}$&5\\
	$o\rightarrow o \leftarrow o$  & $\displaystyle{{ n_c^v}\choose{2}}-\mv{v}{c}{c}-\mlv{v}{c}{c}$&6\\
\end{tabular}~~~~&~~~~
\begin{tabular}{lcc}
	\hline
	{\bf Pattern~~~~~}  & {\bf Frequency}& {\bf Line} \\
	\hline
	\vspace{.3cm}
	 \begin{picture}(0,0)(5,5)
	 \node[Nh=2,Nw=2,Nmr=2,linecolor=white](0)(10,3){$o$}
	 \node[Nh=2,Nw=2,Nmr=2,linecolor=white](1)(15,6){$o$}
	 \node[Nh=2,Nw=2,Nmr=2,linecolor=white](2)(20,3){$o$}
	 \drawedge[ATnb=1,AHnb=1](0,1){}
	 \drawedge[ATnb=1,AHnb=1](1,2){}
	 \drawedge[ATnb=1,AHnb=1](2,0){}
\end{picture}  & $\mlv{v}{a}{a}$&7\\
	\vspace{.3cm}
	 \begin{picture}(0,0)(5,5)
	 \node[Nh=2,Nw=2,Nmr=2,linecolor=white](0)(10,3){$o$}
	 \node[Nh=2,Nw=2,Nmr=2,linecolor=white](1)(15,6){$o$}
	 \node[Nh=2,Nw=2,Nmr=2,linecolor=white](2)(20,3){$o$}
	 \drawedge[ATnb=1,AHnb=0](0,1){}
	 \drawedge[ATnb=1,AHnb=0](1,2){}
	 \drawedge[ATnb=1,AHnb=0](2,0){}
\end{picture}  & $\mv{v}{b}{c}$&8\\
	\vspace{.3cm}
	 \begin{picture}(0,0)(5,5)
	 \node[Nh=2,Nw=2,Nmr=2,linecolor=white](0)(10,3){$o$}
	 \node[Nh=2,Nw=2,Nmr=2,linecolor=white](1)(15,6){$o$}
	 \node[Nh=2,Nw=2,Nmr=2,linecolor=white](2)(20,3){$o$}
	 \drawedge[ATnb=0,AHnb=1](0,1){}
	 \drawedge[ATnb=1,AHnb=0](1,2){}
	 \drawedge[ATnb=1,AHnb=1](2,0){}
\end{picture} & $\displaystyle \mv{v}{a}{b}+\mlv{v}{c}{c}$&9\\
	\vspace{.3cm}
	 \begin{picture}(0,0)(5,5)
	 \node[Nh=2,Nw=2,Nmr=2,linecolor=white](0)(10,3){$o$}
	 \node[Nh=2,Nw=2,Nmr=2,linecolor=white](1)(15,6){$o$}
	 \node[Nh=2,Nw=2,Nmr=2,linecolor=white](2)(20,3){$o$}
	 \drawedge[ATnb=1,AHnb=0](0,1){}
	 \drawedge[ATnb=1,AHnb=0](1,2){}
	 \drawedge[ATnb=1,AHnb=1](2,0){}
\end{picture}  & $\displaystyle \mv{v}{b}{a}+\mv{v}{a}{c}+\mlv{v}{c}{b}$&10\\
	\vspace{.3cm}
	 \begin{picture}(0,0)(5,5)
	 \node[Nh=2,Nw=2,Nmr=2,linecolor=white](0)(10,3){$o$}
	 \node[Nh=2,Nw=2,Nmr=2,linecolor=white](1)(15,6){$o$}
	 \node[Nh=2,Nw=2,Nmr=2,linecolor=white](2)(20,3){$o$}
	 \drawedge[ATnb=1,AHnb=1](0,1){}
	 \drawedge[ATnb=1,AHnb=1](1,2){}
	 \drawedge[ATnb=0,AHnb=1](2,0){}
\end{picture}  & $\displaystyle \mlv{v}{a}{b}+\mlv{v}{a}{c}+\mv{v}{a}{a}$&11\\
	\vspace{.3cm}
	 \begin{picture}(0,0)(5,5)
	 \node[Nh=2,Nw=2,Nmr=2,linecolor=white](0)(10,3){$o$}
	 \node[Nh=2,Nw=2,Nmr=2,linecolor=white](1)(15,6){$o$}
	 \node[Nh=2,Nw=2,Nmr=2,linecolor=white](2)(20,3){$o$}
	 \drawedge[ATnb=1,AHnb=0](0,1){}
	 \drawedge[ATnb=0,AHnb=1](1,2){}
	 \drawedge[ATnb=1,AHnb=1](2,0){}
\end{picture}  & $\displaystyle \mv{v}{c}{a}+\mlv{v}{b}{b}$&12\\
	\vspace{.3cm}
	 \begin{picture}(0,0)(5,5)
	 \node[Nh=2,Nw=2,Nmr=2,linecolor=white](0)(10,3){$o$}
	 \node[Nh=2,Nw=2,Nmr=2,linecolor=white](1)(15,6){$o$}
	 \node[Nh=2,Nw=2,Nmr=2,linecolor=white](2)(20,3){$o$}
	 \drawedge[ATnb=1,AHnb=0](0,1){}
	 \drawedge[ATnb=0,AHnb=1](1,2){}
	 \drawedge[ATnb=1,AHnb=0](2,0){}
\end{picture}  & $\displaystyle \mv{v}{c}{b}+\mv{v}{b}{b}+\mv{v}{c}{c}$&13\\	
\end{tabular}
\end{tabular}
\caption{Isomorphic pattern frequencies involving vertex $v$ and two neighbors.}
\label{table:starb}
\end{center}
\end{table}

 \begin{table}[h!]
\begin{center}
\begin{tabular}{lrr}
	\hline
	\textbf{Pattern~~~~~}  & {\bf Frequency} \\
	\hline
	$\vspace{.2cm} P_4$  & $\displaystyle \size{x} \size{y}^e-\mv{e}{x}{y}$\\
	$\vspace{.2cm} x \leftrightarrow C_3$  & $\displaystyle (\size{x}^e + \size{y}^e) \size{z}^e-\mv{e}{x}{z}-\mv{e}{y}{z}+\mv{e}{x}{x}+\mv{e}{y}{y}$\\
	$ S_3$  & $\displaystyle {{\size{x}^e}\choose{2}}+{{\size{y}^e}\choose{2}}-\mv{e}{x}{x}-\mv{e}{y}{y}$\\
	$K_4\setminus\{e\}$  & $\displaystyle {{\size{z}^e}\choose{2}}+\mv{e}{x}{z}+\m{y}{z}-\mv{e}{z}{z}$\\
	 $C_4$ & $\mv{e}{x}{y}$\\
	 $K_4$ & $\mv{e}{z}{z}$\\
	\hline
\end{tabular}
\caption{$\ept$ frequencies for $e=\{u,v\}$.}
\label{tab:ept}
\end{center}
\end{table}

\begin{table}[h!]
\include{tableQdesori}
\caption{ If the top blue edge is ignored, this table represents $\pt{T_i}{T_j}$ for all $T_i,T_j\in \sett^e$. Symmetric  side omitted.}
\label{table:P}
\end{table}


\begin{table}[htb]
\centerline{
\setlength{\unitlength}{.2mm}
\begin{tabular}{cl}
{\bf Pattern~~~~~}  & {\bf Frequency} \\
\begin{picture}(60,35)(0,00)
\node[Nh=7,Nw=7,Nmr=15](P1)(22,7){} 
\node[Nh=7,Nw=7,Nmr=15](P2)(52,7){} 
\node[Nh=7,Nw=7,Nmr=15](P3)(37,-13){}  
\node[Nh=5,Nw=5,Nmr=15](V1)(15,-13){}  
\node[Nh=5,Nw=5,Nmr=15](V2)(26,-13){}  
  \drawedge[ATnb=0,AHnb=0](P1,P2){}  
   \drawedge[ATnb=0,AHnb=0](P1,P3){}  
   \drawedge[ATnb=0,AHnb=0](P2,P3){}  
   \drawedge[ATnb=0,AHnb=0](P1,V1){}  
      \drawedge[ATnb=0,AHnb=0](P1,V2){}  
   \end{picture} & ${\size{x_1}\choose{2}}+{\size{x_2}\choose{2}}+{\size{x_3}\choose{2}}$\\
\begin{picture}(60,35)(0,0)
\node[Nh=7,Nw=7,Nmr=15](P1)(22,7){} 
\node[Nh=7,Nw=7,Nmr=15](P2)(52,7){} 
\node[Nh=7,Nw=7,Nmr=15](P3)(37,-13){}  
\node[Nh=5,Nw=5,Nmr=15](V1)(52,-13){}  
\node[Nh=5,Nw=5,Nmr=15](V2)(22,-13){}  
  \drawedge[ATnb=0,AHnb=0](P1,P2){}  
   \drawedge[ATnb=0,AHnb=0](P1,P3){}  
   \drawedge[ATnb=0,AHnb=0](P2,P3){}  
   \drawedge[ATnb=0,AHnb=0](P2,V1){}  
      \drawedge[ATnb=0,AHnb=0](P1,V2){}  
   \end{picture} & $\size{x_1}\size{x_2}+\size{x_1}\size{x_3}+\size{x_2}\size{x_3}$\\
\begin{picture}(60,35)(0,0)
\node[Nh=7,Nw=7,Nmr=15](P1)(22,7){} 
\node[Nh=7,Nw=7,Nmr=15](P2)(52,7){} 
\node[Nh=7,Nw=7,Nmr=15](P3)(37,-13){}  
\node[Nh=5,Nw=5,Nmr=15](V1)(52,-13){}  
\node[Nh=5,Nw=5,Nmr=15](V2)(22,-13){}  
  \drawedge[ATnb=0,AHnb=0](P1,P2){}  
   \drawedge[ATnb=0,AHnb=0](P1,P3){}  
   \drawedge[ATnb=0,AHnb=0](P2,P3){}  
   \drawedge[ATnb=0,AHnb=0](P2,V1){}  
   \drawedge[ATnb=0,AHnb=0](P3,V1){}     
      \drawedge[ATnb=0,AHnb=0](P3,V2){}  
   \end{picture} & $\size{x_{1}}(\size{y_{12}}+\size{y_{13}})+\size{x_2}(\size{y_{12}}+\size{y_{23}})+n_{x_3}(\size{y_{13}}+\size{y_{23}})$\\
\begin{picture}(60,35)(0,0)
\node[Nh=7,Nw=7,Nmr=15](P1)(22,7){} 
\node[Nh=7,Nw=7,Nmr=15](P2)(52,7){} 
\node[Nh=7,Nw=7,Nmr=15](P3)(37,-13){}  
\node[Nh=5,Nw=5,Nmr=15](V1)(52,-13){}  
\node[Nh=5,Nw=5,Nmr=15](V2)(22,-13){}  
  \drawedge[ATnb=0,AHnb=0](P1,P2){}  
   \drawedge[ATnb=0,AHnb=0](P1,P3){}  
   \drawedge[ATnb=0,AHnb=0](P2,P3){}  
   \drawedge[ATnb=0,AHnb=0](P2,V1){}  
   \drawedge[ATnb=0,AHnb=0](P3,V1){}     
      \drawedge[ATnb=0,AHnb=0](P1,V2){}  
   \end{picture} & $\size{x_{1}}\size{y_{23}}+\size{x_2}\size{y_{13}}+n_{x_3}\size{y_{23}}$\\
   \begin{picture}(60,35)(0,0)
\node[Nh=7,Nw=7,Nmr=15](P1)(22,7){} 
\node[Nh=7,Nw=7,Nmr=15](P2)(52,7){} 
\node[Nh=7,Nw=7,Nmr=15](P3)(37,-13){}  
\node[Nh=5,Nw=5,Nmr=15](V1)(52,-13){}  
\node[Nh=5,Nw=5,Nmr=15](V2)(37,0){}  
  \drawedge[ATnb=0,AHnb=0](P1,P2){}  
   \drawedge[ATnb=0,AHnb=0](P1,P3){}  
   \drawedge[ATnb=0,AHnb=0](P2,P3){}  
   \drawedge[ATnb=0,AHnb=0](P2,V1){}  
      \drawedge[ATnb=0,AHnb=0](P1,V2){}  
         \drawedge[ATnb=0,AHnb=0](P2,V2){}  
            \drawedge[ATnb=0,AHnb=0](P3,V2){}  
   \end{picture} & $(\size{x_{1}}+\size{x_{2}}+\size{x_{3}})\size{z}$\\
   \begin{picture}(60,35)(0,0)
\node[Nh=7,Nw=7,Nmr=15](P1)(22,7){} 
\node[Nh=7,Nw=7,Nmr=15](P2)(52,7){} 
\node[Nh=7,Nw=7,Nmr=15](P3)(37,-13){}  
\node[Nh=5,Nw=5,Nmr=15](V1)(22,-13){}  
\node[Nh=5,Nw=5,Nmr=15](V2)(37,0){}  
  \drawedge[ATnb=0,AHnb=0](P1,P2){}  
   \drawedge[ATnb=0,AHnb=0](P1,P3){}  
   \drawedge[ATnb=0,AHnb=0](P2,P3){}  
   \drawedge[ATnb=0,AHnb=0](P1,V1){}  
      \drawedge[ATnb=0,AHnb=0](P3,V1){}  
      \drawedge[ATnb=0,AHnb=0](P1,V2){}  
            \drawedge[ATnb=0,AHnb=0](P3,V2){}  
   \end{picture} &  ${\size{y_{12}}\choose{2}}+{\size{y_{13}}\choose{2}}+{\size{y_{23}}\choose{2}}$\\
   \begin{picture}(60,35)(0,0)
\node[Nh=7,Nw=7,Nmr=15](P1)(22,7){} 
\node[Nh=7,Nw=7,Nmr=15](P2)(52,7){} 
\node[Nh=7,Nw=7,Nmr=15](P3)(37,-13){}  
\node[Nh=5,Nw=5,Nmr=15](V1)(52,-13){}  
\node[Nh=5,Nw=5,Nmr=15](V2)(22,-13){}  
  \drawedge[ATnb=0,AHnb=0](P1,P2){}  
   \drawedge[ATnb=0,AHnb=0](P1,P3){}  
   \drawedge[ATnb=0,AHnb=0](P2,P3){}  
   \drawedge[ATnb=0,AHnb=0](P2,V1){}  
      \drawedge[ATnb=0,AHnb=0](P3,V1){}  
      \drawedge[ATnb=0,AHnb=0](P1,V2){}  
      \drawedge[ATnb=0,AHnb=0](P3,V2){}  
   \end{picture} & $\size{y_{12}}\size{y_{13}}+\size{y_{12}}\size{y_{23}}+\size{y_{13}}\size{y_{23}}$\\
     \begin{picture}(60,35)(0,0)
\node[Nh=7,Nw=7,Nmr=15](P1)(22,7){} 
\node[Nh=7,Nw=7,Nmr=15](P2)(52,7){} 
\node[Nh=7,Nw=7,Nmr=15](P3)(37,-13){}  
\node[Nh=5,Nw=5,Nmr=15](V1)(52,-13){}  
\node[Nh=5,Nw=5,Nmr=15](V2)(37,0){}  
  \drawedge[ATnb=0,AHnb=0](P1,P2){}  
   \drawedge[ATnb=0,AHnb=0](P1,P3){}  
   \drawedge[ATnb=0,AHnb=0](P2,P3){}  
   \drawedge[ATnb=0,AHnb=0](P2,V1){}  
   \drawedge[ATnb=0,AHnb=0](P3,V1){}  
      \drawedge[ATnb=0,AHnb=0](P1,V2){}  
         \drawedge[ATnb=0,AHnb=0](P2,V2){}  
            \drawedge[ATnb=0,AHnb=0](P3,V2){}  
   \end{picture} & $(\size{y_{12}}+\size{y_{13}}+\size{y_{23}})\size{z}$\\
     \begin{picture}(60,35)(0,0)
\node[Nh=7,Nw=7,Nmr=15](P1)(22,7){} 
\node[Nh=7,Nw=7,Nmr=15](P2)(52,7){} 
\node[Nh=7,Nw=7,Nmr=15](P3)(37,-13){}  
\node[Nh=5,Nw=5,Nmr=15](V1)(52,-13){}  
\node[Nh=5,Nw=5,Nmr=15](V2)(37,0){}  
  \drawedge[ATnb=0,AHnb=0](P1,P2){}  
   \drawedge[ATnb=0,AHnb=0](P1,P3){}  
   \drawedge[ATnb=0,AHnb=0](P2,P3){}  
   \drawedge[ATnb=0,AHnb=0](P2,V1){}  
   \drawedge[ATnb=0,AHnb=0](P3,V1){}  
      \drawedge[ATnb=0,AHnb=0](P1,V2){}  
         \drawedge[ATnb=0,AHnb=0](P2,V2){}  
            \drawedge[ATnb=0,AHnb=0](P3,V2){}  
            \drawbpedge[ATnb=0,AHnb=0](P1,20,30,V1,60,40){}
   \end{picture} & $\size{z}\choose{2}$\\
   \end{tabular}
}
\caption{$\kpt$ frequencies for the graph of Figure \ref{figure:starfive}. $\ppt$ frequencies are analogous.}
\label{table:starfive}
\end{table}

\begin{table}[htb]
\centerline{
\setlength{\unitlength}{.2mm}
\begin{tabular}{clc}
{\bf Pattern~~~~~}  & {\bf Frequency} & Line \\
\begin{picture}(60,35)(0,00)
\node[Nh=7,Nw=7,Nmr=15](P1)(22,7){} 
\node[Nh=7,Nw=7,Nmr=15](P2)(52,7){} 
\node[Nh=7,Nw=7,Nmr=15](P3)(37,-13){}  
\node[Nh=7,Nw=7,Nmr=15](V1)(15,-13){}  
\node[Nh=7,Nw=7,Nmr=15](V2)(26,-13){}  
  \drawedge[ATnb=0,AHnb=0](P1,P2){}  
   \drawedge[ATnb=0,AHnb=0](P1,P3){}  
   \drawedge[ATnb=0,AHnb=0](P2,P3){}  
   \drawedge[ATnb=0,AHnb=0](P1,V1){}  
      \drawedge[ATnb=0,AHnb=0](P1,V2){}  
   \end{picture} & ${\size{x_1}\choose{2}}+{\size{x_2}\choose{2}}+{\size{x_3}\choose{2}}-\mv{}{x_1}{x_1}-\mv{}{x_2}{x_2}-\mv{}{x_3}{x_3}$ &\nextdoc\\
\begin{picture}(60,35)(0,0)
\node[Nh=7,Nw=7,Nmr=15](P1)(22,7){} 
\node[Nh=7,Nw=7,Nmr=15](P2)(52,7){} 
\node[Nh=7,Nw=7,Nmr=15](P3)(37,-13){}  
\node[Nh=7,Nw=7,Nmr=15](V1)(52,-13){}  
\node[Nh=7,Nw=7,Nmr=15](V2)(22,-13){}  
  \drawedge[ATnb=0,AHnb=0](P1,P2){}  
   \drawedge[ATnb=0,AHnb=0](P1,P3){}  
   \drawedge[ATnb=0,AHnb=0](P2,P3){}  
   \drawedge[ATnb=0,AHnb=0](P2,V1){}  
      \drawedge[ATnb=0,AHnb=0](P1,V2){}  
   \end{picture} & $\size{x_1}\size{x_2}+\size{x_1}\size{x_3}+\size{x_2}\size{x_3}-\mv{}{x_1}{x_2}-\mv{}{x_1}{x_3}-\mv{}{x_2}{x_3}$ &\nextdoc\\
\begin{picture}(60,35)(0,0)
\node[Nh=7,Nw=7,Nmr=15](P1)(22,7){} 
\node[Nh=7,Nw=7,Nmr=15](P2)(52,7){} 
\node[Nh=7,Nw=7,Nmr=15](P3)(37,-13){}  
\node[Nh=7,Nw=7,Nmr=15](V1)(52,-13){}  
\node[Nh=7,Nw=7,Nmr=15](V2)(22,-13){}  
  \drawedge[ATnb=0,AHnb=0](P1,P2){}  
   \drawedge[ATnb=0,AHnb=0](P1,P3){}  
   \drawedge[ATnb=0,AHnb=0](P2,P3){}  
   \drawedge[ATnb=0,AHnb=0](P2,V1){}  
   \drawedge[ATnb=0,AHnb=0](P3,V1){}     
      \drawedge[ATnb=0,AHnb=0](P3,V2){}  
   \end{picture} & $\size{x_{1}}(\size{y_{12}}+\size{y_{13}})+\size{x_2}(\size{y_{12}}+\size{y_{23}})+n_{x_3}(\size{y_{13}}+\size{y_{23}})$&\nextdoc\\
   & $-\mv{}{x_1}{y_{12}}-\mv{}{x_1}{y_{13}}-\mv{}{x_2}{y_{12}}-\mv{}{x_2}{y_{23}}-\mv{}{x_3}{y_{13}}-\mv{}{x_3}{y_{23}}$\\
\begin{picture}(60,35)(0,0)
\node[Nh=7,Nw=7,Nmr=15](P1)(22,7){} 
\node[Nh=7,Nw=7,Nmr=15](P2)(52,7){} 
\node[Nh=7,Nw=7,Nmr=15](P3)(37,-13){}  
\node[Nh=7,Nw=7,Nmr=15](V1)(52,-13){}  
\node[Nh=7,Nw=7,Nmr=15](V2)(22,-13){}  
  \drawedge[ATnb=0,AHnb=0](P1,P2){}  
   \drawedge[ATnb=0,AHnb=0](P1,P3){}  
   \drawedge[ATnb=0,AHnb=0](P2,P3){}  
   \drawedge[ATnb=0,AHnb=0](P2,V1){}  
   \drawedge[ATnb=0,AHnb=0](P3,V1){}     
      \drawedge[ATnb=0,AHnb=0](P1,V2){}  
   \end{picture} & $\size{x_{1}}\size{y_{23}}+\size{x_2}\size{y_{13}}+n_{x_3}\size{y_{23}}-\mv{}{x_1}{y_{23}}-\mv{}{x_2}{y_{13}}-\mv{}{x_3}{y_{12}}$&\nextdoc\\
   \begin{picture}(60,35)(0,0)
\node[Nh=7,Nw=7,Nmr=15](P1)(22,7){} 
\node[Nh=7,Nw=7,Nmr=15](P2)(52,7){} 
\node[Nh=7,Nw=7,Nmr=15](P3)(37,-13){}  
\node[Nh=7,Nw=7,Nmr=15](V1)(52,-13){}  
\node[Nh=7,Nw=7,Nmr=15](V2)(37,0){}  
  \drawedge[ATnb=0,AHnb=0](P1,P2){}  
   \drawedge[ATnb=0,AHnb=0](P1,P3){}  
   \drawedge[ATnb=0,AHnb=0](P2,P3){}  
   \drawedge[ATnb=0,AHnb=0](P2,V1){}  
      \drawedge[ATnb=0,AHnb=0](P1,V2){}  
         \drawedge[ATnb=0,AHnb=0](P2,V2){}  
            \drawedge[ATnb=0,AHnb=0](P3,V2){}  
   \end{picture} & $(\size{x_{1}}+\size{x_{2}}+\size{x_{3}})\size{z}-\mv{}{x_1}{z}-\mv{}{x_2}{z}-\mv{}{x_3}{z}$&\nextdoc\\
   \begin{picture}(60,35)(0,0)
\node[Nh=7,Nw=7,Nmr=15](P1)(22,7){} 
\node[Nh=7,Nw=7,Nmr=15](P2)(52,7){} 
\node[Nh=7,Nw=7,Nmr=15](P3)(37,-13){}  
\node[Nh=7,Nw=7,Nmr=15](V1)(22,-13){}  
\node[Nh=7,Nw=7,Nmr=15](V2)(37,0){}  
  \drawedge[ATnb=0,AHnb=0](P1,P2){}  
   \drawedge[ATnb=0,AHnb=0](P1,P3){}  
   \drawedge[ATnb=0,AHnb=0](P2,P3){}  
   \drawedge[ATnb=0,AHnb=0](P1,V1){}  
      \drawedge[ATnb=0,AHnb=0](P3,V1){}  
      \drawedge[ATnb=0,AHnb=0](P1,V2){}  
            \drawedge[ATnb=0,AHnb=0](P3,V2){}  
   \end{picture} &  ${\size{y_{12}}\choose{2}}+{\size{y_{13}}\choose{2}}+{\size{y_{23}}\choose{2}}-\mv{}{y_{12}}{y_{13}}-\mv{}{y_{12}}{y_{23}}-\mv{}{y_{13}}{y_{23}}$&\nextdoc\\
   \begin{picture}(60,35)(0,0)
\node[Nh=7,Nw=7,Nmr=15](P1)(22,7){} 
\node[Nh=7,Nw=7,Nmr=15](P2)(52,7){} 
\node[Nh=7,Nw=7,Nmr=15](P3)(37,-13){}  
\node[Nh=7,Nw=7,Nmr=15](V1)(52,-13){}  
\node[Nh=7,Nw=7,Nmr=15](V2)(22,-13){}  
  \drawedge[ATnb=0,AHnb=0](P1,P2){}  
   \drawedge[ATnb=0,AHnb=0](P1,P3){}  
   \drawedge[ATnb=0,AHnb=0](P2,P3){}  
   \drawedge[ATnb=0,AHnb=0](P2,V1){}  
      \drawedge[ATnb=0,AHnb=0](P3,V1){}  
      \drawedge[ATnb=0,AHnb=0](P1,V2){}  
      \drawedge[ATnb=0,AHnb=0](P3,V2){}  
   \end{picture} & $\size{y_{12}}\size{y_{13}}+\size{y_{12}}\size{y_{23}}+\size{y_{13}}\size{y_{23}}-\mv{}{y_{12}}{y_{13}}-\mv{}{y_{12}}{y_{23}}-\mv{}{y_{13}}{\size{y_{23}}}$&\nextdoc\\
   & +$\mv{}{x_1}{y_{12}}+\mv{}{x_1}{y_{13}}+\mv{}{x_2}{y_{12}}+\mv{}{x_2}{y_{23}}+\mv{}{x_3}{y_{13}}+\mv{}{x_3}{y_{23}}$\\
     \begin{picture}(60,35)(0,0)
\node[Nh=7,Nw=7,Nmr=15](P1)(22,7){} 
\node[Nh=7,Nw=7,Nmr=15](P2)(52,7){} 
\node[Nh=7,Nw=7,Nmr=15](P3)(37,-13){}  
\node[Nh=7,Nw=7,Nmr=15](V1)(52,-13){}  
\node[Nh=7,Nw=7,Nmr=15](V2)(37,0){}  
  \drawedge[ATnb=0,AHnb=0](P1,P2){}  
   \drawedge[ATnb=0,AHnb=0](P1,P3){}  
   \drawedge[ATnb=0,AHnb=0](P2,P3){}  
   \drawedge[ATnb=0,AHnb=0](P2,V1){}  
   \drawedge[ATnb=0,AHnb=0](P3,V1){}  
      \drawedge[ATnb=0,AHnb=0](P1,V2){}  
         \drawedge[ATnb=0,AHnb=0](P2,V2){}  
            \drawedge[ATnb=0,AHnb=0](P3,V2){}  
   \end{picture} & $(\size{y_{12}}+\size{y_{13}}+\size{y_{23}})\size{z}-\mv{}{y_{12}}{z}-\mv{}{y_{13}}{z}-\mv{}{y_{23}}{z}$&\nextdoc\\
      &$\mv{}{x_1}{z}+\mv{}{x_2}{z}+\mv{}{x_3}{z}+\mv{}{y_{12}}{y_{13}}+\mv{}{y_{12}}{y_{23}}+\mv{}{y_{13}}{y_{23}}$\\
       \begin{picture}(60,35)(0,0)
\node[Nh=7,Nw=7,Nmr=15](P1)(22,7){} 
\node[Nh=7,Nw=7,Nmr=15](P2)(52,7){} 
\node[Nh=7,Nw=7,Nmr=15](P3)(37,-13){}  
\node[Nh=7,Nw=7,Nmr=15](V1)(52,-13){}  
\node[Nh=7,Nw=7,Nmr=15](V2)(37,0){}  
  \drawedge[ATnb=0,AHnb=0](P1,P2){}  
   \drawedge[ATnb=0,AHnb=0](P1,P3){}  
   \drawedge[ATnb=0,AHnb=0](P2,P3){}  
   \drawedge[ATnb=0,AHnb=0](P2,V1){}  
   \drawedge[ATnb=0,AHnb=0](P3,V1){}  
      \drawedge[ATnb=0,AHnb=0](P1,V2){}  
         \drawedge[ATnb=0,AHnb=0](P2,V2){}  
            \drawedge[ATnb=0,AHnb=0](V1,V2){}  
   \end{picture} & $\mv{}{\size{y_{12}}}{\size{y_{13}}}+\mv{}{\size{y_{12}}}{\size{y_{23}}}+\mv{}{\size{y_{13}}}{\size{y_{23}}}$&\nextdoc\\
     \begin{picture}(60,35)(0,0)
\node[Nh=7,Nw=7,Nmr=15](P1)(22,7){} 
\node[Nh=7,Nw=7,Nmr=15](P2)(52,7){} 
\node[Nh=7,Nw=7,Nmr=15](P3)(37,-13){}  
\node[Nh=7,Nw=7,Nmr=15](V1)(52,-13){}  
\node[Nh=7,Nw=7,Nmr=15](V2)(37,0){}  
  \drawedge[ATnb=0,AHnb=0](P1,P2){}  
   \drawedge[ATnb=0,AHnb=0](P1,P3){}  
   \drawedge[ATnb=0,AHnb=0](P2,P3){}  
   \drawedge[ATnb=0,AHnb=0](P2,V1){}  
   \drawedge[ATnb=0,AHnb=0](P3,V1){}  
      \drawedge[ATnb=0,AHnb=0](P1,V2){}  
         \drawedge[ATnb=0,AHnb=0](P2,V2){}  
            \drawedge[ATnb=0,AHnb=0](P3,V2){}  
            \drawbpedge[ATnb=0,AHnb=0](P1,20,30,V1,60,40){}
   \end{picture} & ${\size{z}\choose{2}}-m{}{z}{z}+\mv{}{y_{12}}{z}+\mv{}{y_{13}}{z}+\mv{}{y_{23}}{z}$&\nextdoc\\
   \begin{picture}(60,35)(0,00)
\node[Nh=7,Nw=7,Nmr=15](P1)(22,7){} 
\node[Nh=7,Nw=7,Nmr=15](P2)(52,7){} 
\node[Nh=7,Nw=7,Nmr=15](P3)(37,-13){}  
\node[Nh=7,Nw=7,Nmr=15](V1)(15,-13){}  
\node[Nh=7,Nw=7,Nmr=15](V2)(26,-13){}  
  \drawedge[ATnb=0,AHnb=0](P1,P2){}  
   \drawedge[ATnb=0,AHnb=0](P1,P3){}  
   \drawedge[ATnb=0,AHnb=0](P2,P3){}  
   \drawedge[ATnb=0,AHnb=0](P1,V1){}  
      \drawedge[ATnb=0,AHnb=0](P1,V2){} 
      \drawedge[ATnb=0,AHnb=0](V1,V2){}  
   \end{picture} & $\mv{}{x_1}{x_1}+\mv{}{x_2}{x_2}+\mv{}{x_3}{x_3}$&\nextdoc\\
      \begin{picture}(60,35)(0,00)
\node[Nh=7,Nw=7,Nmr=15](P1)(22,7){} 
\node[Nh=7,Nw=7,Nmr=15](P2)(52,7){} 
\node[Nh=7,Nw=7,Nmr=15](P3)(37,-13){}  
\node[Nh=7,Nw=7,Nmr=15](V1)(15,-13){}  
\node[Nh=7,Nw=7,Nmr=15](V2)(26,-13){}  
  \drawedge[ATnb=0,AHnb=0](P1,P2){}  
   \drawedge[ATnb=0,AHnb=0](P1,P3){}  
   \drawedge[ATnb=0,AHnb=0](P2,P3){}  
   \drawedge[ATnb=0,AHnb=0](P1,V1){}  
      \drawedge[ATnb=0,AHnb=0](P3,V2){} 
      \drawedge[ATnb=0,AHnb=0](V1,V2){}  
   \end{picture} & $\mv{}{x_1}{x_2}+\mv{}{x_1}{x_3}+\mv{}{x_2}{x_3}$&\nextdoc\\
     \begin{picture}(60,35)(0,0)
\node[Nh=7,Nw=7,Nmr=15](P1)(22,7){} 
\node[Nh=7,Nw=7,Nmr=15](P2)(52,7){} 
\node[Nh=7,Nw=7,Nmr=15](P3)(37,-13){}  
\node[Nh=7,Nw=7,Nmr=15](V1)(22,-13){}  
\node[Nh=7,Nw=7,Nmr=15](V2)(37,0){}  
  \drawedge[ATnb=0,AHnb=0](P1,P2){}  
   \drawedge[ATnb=0,AHnb=0](V1,V2){}  
   \drawedge[ATnb=0,AHnb=0](P2,P3){}  
   \drawedge[ATnb=0,AHnb=0](P1,V1){}  
      \drawedge[ATnb=0,AHnb=0](P3,V1){}  
      \drawedge[ATnb=0,AHnb=0](P1,V2){}  
            \drawedge[ATnb=0,AHnb=0](P3,V2){}  
   \end{picture} &  $\mv{}{x_1}{y_{23}}+\mv{}{x_2}{y_{13}}+\mv{}{x_3}{y_{12}}$&\nextdoc\\
    \begin{picture}(60,35)(0,0)
\node[Nh=7,Nw=7,Nmr=15](P1)(22,7){} 
\node[Nh=7,Nw=7,Nmr=15](P2)(52,7){} 
\node[Nh=7,Nw=7,Nmr=15](P3)(37,-13){}  
\node[Nh=7,Nw=7,Nmr=15](V1)(52,-13){}  
\node[Nh=7,Nw=7,Nmr=15](V2)(37,0){}  
  \drawedge[ATnb=0,AHnb=0](P1,P2){}  
   \drawedge[ATnb=0,AHnb=0](P1,P3){}  
   \drawedge[ATnb=0,AHnb=0](P2,P3){}  
   \drawedge[ATnb=0,AHnb=0](P2,V1){}  
   \drawedge[ATnb=0,AHnb=0](P3,V1){}  
   \drawedge[ATnb=0,AHnb=0](V1,V2){}     
      \drawedge[ATnb=0,AHnb=0](P1,V2){}  
         \drawedge[ATnb=0,AHnb=0](P2,V2){}  
            \drawedge[ATnb=0,AHnb=0](P3,V2){}  
            \drawbpedge[ATnb=0,AHnb=0](P1,20,30,V1,60,40){}
   \end{picture} & $\mv{}{z}{z}$&\nextdoc\\
   \end{tabular}
}
\caption{$\kpt$ frequencies in a general graph.}
\label{table:histok}
\end{table}

\begin{table}[htb]
\centerline{
\setlength{\unitlength}{.2mm}
\begin{tabular}{clcc}
{\bf Pattern~~~~~}  & {\bf Frequency}  \\
\begin{picture}(60,35)(0,00)
\node[Nh=7,Nw=7,Nmr=15](P1)(22,7){} 
\node[Nh=7,Nw=7,Nmr=15](P2)(52,7){} 
\node[Nh=7,Nw=7,Nmr=15](P3)(37,-13){}  
\node[Nh=7,Nw=7,Nmr=15](V1)(15,-13){}  
\node[Nh=7,Nw=7,Nmr=15](V2)(26,-13){}  
   \drawedge[ATnb=0,AHnb=0](P1,P3){}  
   \drawedge[ATnb=0,AHnb=0](P2,P3){}  
   \drawedge[ATnb=0,AHnb=0](P1,V1){}  
      \drawedge[ATnb=0,AHnb=0](P1,V2){}  
   \end{picture} & ${\size{x_1}\choose{2}}+{\size{x_2}\choose{2}}-\mv{}{x_1}{x_1}-\mv{}{x_3}{x_3}$\\
   \begin{picture}(60,35)(0,00)
\node[Nh=7,Nw=7,Nmr=15](P1)(22,7){} 
\node[Nh=7,Nw=7,Nmr=15](P2)(52,7){} 
\node[Nh=7,Nw=7,Nmr=15](P3)(37,-13){}  
\node[Nh=7,Nw=7,Nmr=15](V1)(15,-13){}  
\node[Nh=7,Nw=7,Nmr=15](V2)(26,-13){}  
  \drawedge[ATnb=0,AHnb=0](P1,P2){}  
   \drawedge[ATnb=0,AHnb=0](P1,P3){}  
   \drawedge[ATnb=0,AHnb=0](P1,V1){}  
      \drawedge[ATnb=0,AHnb=0](P1,V2){}  
   \end{picture} & ${\size{x_2}\choose{2}}-\mv{}{x_2}{x_2}-\mv{}{x_3}{x_3}$\\
\begin{picture}(60,35)(0,0)
\node[Nh=7,Nw=7,Nmr=15](P1)(22,7){} 
\node[Nh=7,Nw=7,Nmr=15](P2)(52,7){} 
\node[Nh=7,Nw=7,Nmr=15](P3)(37,-13){}  
\node[Nh=7,Nw=7,Nmr=15](V1)(52,-13){}  
\node[Nh=7,Nw=7,Nmr=15](V2)(22,-13){}  
  \drawedge[ATnb=0,AHnb=0](P1,P2){}  
   \drawedge[ATnb=0,AHnb=0](P2,P3){}  
   \drawedge[ATnb=0,AHnb=0](P2,V1){}  
      \drawedge[ATnb=0,AHnb=0](P1,V2){}  
   \end{picture} & $\size{x_1}\size{x_2}+\size{x_2}\size{x_3}-\mv{}{x_1}{x_2}-\mv{}{x_2}{x_3}$\\
\begin{picture}(60,35)(0,0)
\node[Nh=7,Nw=7,Nmr=15](P1)(22,7){} 
\node[Nh=7,Nw=7,Nmr=15](P2)(52,7){} 
\node[Nh=7,Nw=7,Nmr=15](P3)(37,-13){}  
\node[Nh=7,Nw=7,Nmr=15](V1)(52,-13){}  
\node[Nh=7,Nw=7,Nmr=15](V2)(22,-13){}  
   \drawedge[ATnb=0,AHnb=0](P1,P3){}  
   \drawedge[ATnb=0,AHnb=0](P2,P3){}  
   \drawedge[ATnb=0,AHnb=0](P2,V1){}  
      \drawedge[ATnb=0,AHnb=0](P1,V2){}  
   \end{picture} & $\size{x_1}\size{x_3}-\mv{}{x_1}{x_3}$ \\
\begin{picture}(60,35)(0,0)
\node[Nh=7,Nw=7,Nmr=15](P1)(22,7){} 
\node[Nh=7,Nw=7,Nmr=15](P2)(52,7){} 
\node[Nh=7,Nw=7,Nmr=15](P3)(37,-13){}  
\node[Nh=7,Nw=7,Nmr=15](V1)(52,-13){}  
\node[Nh=7,Nw=7,Nmr=15](V2)(22,-13){}  
  \drawedge[ATnb=0,AHnb=0](P1,P2){}  
   \drawedge[ATnb=0,AHnb=0](P1,P3){}  
   \drawedge[ATnb=0,AHnb=0](P2,P3){}  
   \drawedge[ATnb=0,AHnb=0](P2,V1){}  
      \drawedge[ATnb=0,AHnb=0](P1,V2){}  
   \end{picture}  & $\size{x_{1}}\size{y_{12}}+\size{x_3}\size{y_{23}}-\mv{}{x_1}{y_{12}}-\mv{}{x_3}{y_{23}}$\\
\begin{picture}(60,35)(0,0)
\node[Nh=7,Nw=7,Nmr=15](P1)(22,7){} 
\node[Nh=7,Nw=7,Nmr=15](P2)(52,7){} 
\node[Nh=7,Nw=7,Nmr=15](P3)(37,-13){}  
\node[Nh=7,Nw=7,Nmr=15](V1)(52,-13){}  
\node[Nh=7,Nw=7,Nmr=15](V2)(22,-13){}  
  \drawedge[ATnb=0,AHnb=0](P1,P2){}  
   \drawedge[ATnb=0,AHnb=0](P1,P3){}  
   \drawedge[ATnb=0,AHnb=0](P3,V1){}  
   \drawedge[ATnb=0,AHnb=0](P2,V1){}  
      \drawedge[ATnb=0,AHnb=0](P1,V2){}  
   \end{picture}  & $\size{x_{1}}\size{y_{13}}+\size{x_3}\size{y_{13}}+\size{x_2}\size{y_{13}}-\mv{}{x_1}{y_{13}}-\mv{}{x_3}{y_{13}}-\mv{}{x_2}{y_{13}}+\mv{}{x_1}{x_2}+\mv{}{x_2}{x_3}$\\
\begin{picture}(60,35)(0,0)
\node[Nh=7,Nw=7,Nmr=15](P1)(22,7){} 
\node[Nh=7,Nw=7,Nmr=15](P2)(52,7){} 
\node[Nh=7,Nw=7,Nmr=15](P3)(37,-13){}  
\node[Nh=7,Nw=7,Nmr=15](V1)(52,-13){}  
\node[Nh=7,Nw=7,Nmr=15](V2)(22,-13){}  
  \drawedge[ATnb=0,AHnb=0](P2,P3){}  
   \drawedge[ATnb=0,AHnb=0](P1,P3){}  
   \drawedge[ATnb=0,AHnb=0](P3,V1){}  
   \drawedge[ATnb=0,AHnb=0](P2,V1){}  
      \drawedge[ATnb=0,AHnb=0](P1,V2){}  
   \end{picture}  & $\size{x_{1}}\size{y_{23}}+\size{x_3}\size{y_{12}}-\mv{}{x_1}{y_{23}}-\mv{}{x_3}{y_{12}}+\mv{}{x_1}{x_1}+\mv{}{x_3}{x_3}$ \\
\begin{picture}(60,35)(0,00)
\node[Nh=7,Nw=7,Nmr=15](P1)(22,7){} 
\node[Nh=7,Nw=7,Nmr=15](P2)(52,7){} 
\node[Nh=7,Nw=7,Nmr=15](P3)(37,-13){}  
\node[Nh=7,Nw=7,Nmr=15](V1)(15,-13){}  
\node[Nh=7,Nw=7,Nmr=15](V2)(26,-13){}  
  \drawedge[ATnb=0,AHnb=0](P1,P2){}  
   \drawedge[ATnb=0,AHnb=0](P1,P3){}  
   \drawedge[ATnb=0,AHnb=0](P2,P3){}  
   \drawedge[ATnb=0,AHnb=0](P1,V1){}  
      \drawedge[ATnb=0,AHnb=0](P1,V2){}  
   \end{picture}  & $\size{x_{2}}\size{y_{12}}+\size{x_2}\size{y_{23}}-\mv{}{x_2}{y_{12}}-\mv{}{x_2}{y_{23}}+\mv{}{x_2}{x_2}$\\
   \begin{picture}(60,35)(0,0)
\node[Nh=7,Nw=7,Nmr=15](P1)(22,7){} 
\node[Nh=7,Nw=7,Nmr=15](P2)(52,7){} 
\node[Nh=7,Nw=7,Nmr=15](P3)(37,-13){}  
\node[Nh=7,Nw=7,Nmr=15](V1)(52,-13){}  
\node[Nh=7,Nw=7,Nmr=15](V2)(37,0){}  
   \drawedge[ATnb=0,AHnb=0](P3,V1){}  
   \drawedge[ATnb=0,AHnb=0](P2,P3){}  
   \drawedge[ATnb=0,AHnb=0](P2,V1){}  
      \drawedge[ATnb=0,AHnb=0](P1,V2){}  
         \drawedge[ATnb=0,AHnb=0](P2,V2){}  
            \drawedge[ATnb=0,AHnb=0](P3,V2){}  
   \end{picture} & $(\size{x_{1}}+\size{x_{3}})\size{z}-\mv{}{x_1}{z}-\mv{}{x_3}{z}+\mv{}{x_1}{y_{12}}+\mv{}{x_3}{y_{23}}$\\
      \begin{picture}(60,35)(0,0)
\node[Nh=7,Nw=7,Nmr=15](P1)(22,7){} 
\node[Nh=7,Nw=7,Nmr=15](P2)(52,7){} 
\node[Nh=7,Nw=7,Nmr=15](P3)(37,-13){}  
\node[Nh=7,Nw=7,Nmr=15](V1)(52,-13){}  
\node[Nh=7,Nw=7,Nmr=15](V2)(37,0){}  
   \drawedge[ATnb=0,AHnb=0](P3,V1){}  
   \drawedge[ATnb=0,AHnb=0](V1,V2){}  
   \drawedge[ATnb=0,AHnb=0](P2,V1){}  
      \drawedge[ATnb=0,AHnb=0](P1,V2){}  
         \drawedge[ATnb=0,AHnb=0](P2,V2){}  
            \drawedge[ATnb=0,AHnb=0](P3,V2){}  
   \end{picture} & $\size{x_{2}}\size{z}+{{\size{y_{12}}\choose{2}}}+{{\size{y_{23}}\choose{2}}}-\mv{}{x_2}{z}-\mv{}{y_{12}}{y_{12}}-\mv{}{y_{23}}{y_{23}}+\mv{}{x_2}{y_{12}}+\mv{}{x_2}{y_{23}}$\\
  \begin{picture}(60,35)(0,0)
\node[Nh=7,Nw=7,Nmr=15](P1)(22,7){} 
\node[Nh=7,Nw=7,Nmr=15](P2)(52,7){} 
\node[Nh=7,Nw=7,Nmr=15](P3)(37,-13){}  
\node[Nh=7,Nw=7,Nmr=15](V1)(22,-13){}  
\node[Nh=7,Nw=7,Nmr=15](V2)(37,0){}  
  \drawedge[ATnb=0,AHnb=0](P1,P2){}  
   \drawedge[ATnb=0,AHnb=0](P2,P3){}  
   \drawedge[ATnb=0,AHnb=0](P1,V1){}  
      \drawedge[ATnb=0,AHnb=0](P3,V1){}  
      \drawedge[ATnb=0,AHnb=0](P1,V2){}  
            \drawedge[ATnb=0,AHnb=0](P3,V2){}  
   \end{picture}&  ${\size{y_{13}}\choose{2}}-\mv{}{y_{13}}{y_{13}}+\mv{}{x_1}{y_{23}}+\mv{}{x_2}{y_{13}}+\mv{}{x_3}{y_{12}}+\mv{}{x_3}{y_{13}}+\mv{}{12}{13}$\\
   \begin{picture}(60,35)(0,00)
\node[Nh=7,Nw=7,Nmr=15](P1)(22,7){} 
\node[Nh=7,Nw=7,Nmr=15](P2)(52,7){} 
\node[Nh=7,Nw=7,Nmr=15](P3)(37,-13){}  
\node[Nh=7,Nw=7,Nmr=15](V1)(15,-13){}  
\node[Nh=7,Nw=7,Nmr=15](V2)(26,-13){}  
  \drawedge[ATnb=0,AHnb=0](P1,P2){}  
   \drawedge[ATnb=0,AHnb=0](P1,P3){}  
   \drawedge[ATnb=0,AHnb=0](P2,P3){}  
   \drawedge[ATnb=0,AHnb=0](P1,V1){}  
      \drawedge[ATnb=0,AHnb=0](P3,V2){} 
      \drawedge[ATnb=0,AHnb=0](V1,V2){}  
   \end{picture}  & $\size{y_{13}}\size{y_{23}}+\size{y_{12}}\size{y_{13}}-\mv{}{y_{13}}{y_{23}}-\mv{}{y_{12}}{y_{13}}$\\
     \begin{picture}(60,35)(0,00)
\node[Nh=7,Nw=7,Nmr=15](P1)(22,7){} 
\node[Nh=7,Nw=7,Nmr=15](P2)(52,7){} 
\node[Nh=7,Nw=7,Nmr=15](P3)(37,-13){}  
\node[Nh=7,Nw=7,Nmr=15](V1)(15,-13){}  
\node[Nh=7,Nw=7,Nmr=15](V2)(26,-13){}  
  \drawedge[ATnb=0,AHnb=0](P1,P2){}  
   \drawedge[ATnb=0,AHnb=0](P1,P3){}  
   \drawedge[ATnb=0,AHnb=0](P2,P3){}  
   \drawedge[ATnb=0,AHnb=0](P1,V1){}  
      \drawedge[ATnb=0,AHnb=0](P1,V2){} 
      \drawedge[ATnb=0,AHnb=0](V1,V2){}  
   \end{picture} & $\size{y_{12}}\size{y_{23}}-\mv{}{y_{12}}{y_{23}}$ \\
       \begin{picture}(60,35)(0,0)
\node[Nh=7,Nw=7,Nmr=15](P1)(22,7){} 
\node[Nh=7,Nw=7,Nmr=15](P2)(52,7){} 
\node[Nh=7,Nw=7,Nmr=15](P3)(37,-13){}  
\node[Nh=7,Nw=7,Nmr=15](V1)(52,-13){}  
\node[Nh=7,Nw=7,Nmr=15](V2)(22,-13){}  
  \drawedge[ATnb=0,AHnb=0](P1,P2){}  
   \drawedge[ATnb=0,AHnb=0](P1,P3){}  
   \drawedge[ATnb=0,AHnb=0](P2,P3){}  
   \drawedge[ATnb=0,AHnb=0](P2,V1){}  
      \drawedge[ATnb=0,AHnb=0](P3,V1){}  
      \drawedge[ATnb=0,AHnb=0](P1,V2){}  
      \drawedge[ATnb=0,AHnb=0](P3,V2){}  
   \end{picture} & $\size{y_{12}}\size{z}+\size{y_{23}}\size{z}-\mv{}{y_{12}}{z}-\mv{}{y_{23}}{z}+\mv{}{x_1}{z}+\mv{}{x_3}{z}+\mv{}{y_{12}}{y_{23}}$\\
     \begin{picture}(60,35)(0,0)
\node[Nh=7,Nw=7,Nmr=15](P1)(22,7){} 
\node[Nh=7,Nw=7,Nmr=15](P2)(52,7){} 
\node[Nh=7,Nw=7,Nmr=15](P3)(37,-13){}  
\node[Nh=7,Nw=7,Nmr=15](V1)(22,-13){}  
\node[Nh=7,Nw=7,Nmr=15](V2)(37,0){}  
  \drawedge[ATnb=0,AHnb=0](P1,P2){}  
   \drawedge[ATnb=0,AHnb=0](V1,V2){}  
   \drawedge[ATnb=0,AHnb=0](P2,P3){}  
   \drawedge[ATnb=0,AHnb=0](P1,V1){}  
      \drawedge[ATnb=0,AHnb=0](P3,V1){}  
      \drawedge[ATnb=0,AHnb=0](P1,V2){}  
            \drawedge[ATnb=0,AHnb=0](P3,V2){}  
   \end{picture}& $\size{y_{13}}\size{z}-\mv{y_{13}}{z}+\mv{}{y_{13}}{y_{13}}+\mv{}{13}{23}$\\
    \begin{picture}(60,35)(0,0)
\node[Nh=7,Nw=7,Nmr=15](P1)(22,7){} 
\node[Nh=7,Nw=7,Nmr=15](P2)(52,7){} 
\node[Nh=7,Nw=7,Nmr=15](P3)(37,-13){}  
\node[Nh=7,Nw=7,Nmr=15](V1)(52,-13){}  
\node[Nh=7,Nw=7,Nmr=15](V2)(37,0){}  
  \drawedge[ATnb=0,AHnb=0](P1,P2){}  
   \drawedge[ATnb=0,AHnb=0](P1,P3){}  
   \drawedge[ATnb=0,AHnb=0](P2,P3){}  
   \drawedge[ATnb=0,AHnb=0](P2,V1){}  
   \drawedge[ATnb=0,AHnb=0](P3,V1){}  
      \drawedge[ATnb=0,AHnb=0](P1,V2){}  
         \drawedge[ATnb=0,AHnb=0](P2,V2){}  
            \drawedge[ATnb=0,AHnb=0](V1,V2){}  
   \end{picture} & ${\size{z}\choose{2}}-\mv{}{z}{z}$\\
     \begin{picture}(60,35)(0,0)
\node[Nh=7,Nw=7,Nmr=15](P1)(22,7){} 
\node[Nh=7,Nw=7,Nmr=15](P2)(52,7){} 
\node[Nh=7,Nw=7,Nmr=15](P3)(37,-13){}  
\node[Nh=7,Nw=7,Nmr=15](V1)(52,-13){}  
\node[Nh=7,Nw=7,Nmr=15](V2)(22,-13){}  
  \drawedge[ATnb=0,AHnb=0](P1,P2){}  
   \drawedge[ATnb=0,AHnb=0](P3,V2){}  
   \drawedge[ATnb=0,AHnb=0](P3,V1){}  
   \drawedge[ATnb=0,AHnb=0](P2,V1){}  
      \drawedge[ATnb=0,AHnb=0](P1,V2){}  
   \end{picture}  & $\mv{}{x_1}{x_3}$\\
      \begin{picture}(60,35)(0,0)
\node[Nh=7,Nw=7,Nmr=15](P1)(22,7){} 
\node[Nh=7,Nw=7,Nmr=15](P2)(52,7){} 
\node[Nh=7,Nw=7,Nmr=15](P3)(37,-13){}  
\node[Nh=7,Nw=7,Nmr=15](V1)(22,-13){}  
\node[Nh=7,Nw=7,Nmr=15](V2)(37,0){}  
  \drawedge[ATnb=0,AHnb=0](P1,P2){}  
   \drawedge[ATnb=0,AHnb=0](P1,P3){}  
   \drawedge[ATnb=0,AHnb=0](P2,P3){}  
   \drawedge[ATnb=0,AHnb=0](P1,V1){}  
      \drawedge[ATnb=0,AHnb=0](P3,V1){}  
      \drawedge[ATnb=0,AHnb=0](P1,V2){}  
            \drawedge[ATnb=0,AHnb=0](P3,V2){}  
   \end{picture}& $\mv{}{x_2}{z}$\\
      \begin{picture}(60,35)(0,0)
\node[Nh=7,Nw=7,Nmr=15](P1)(22,7){} 
\node[Nh=7,Nw=7,Nmr=15](P2)(52,7){} 
\node[Nh=7,Nw=7,Nmr=15](P3)(37,-13){}  
\node[Nh=7,Nw=7,Nmr=15](V1)(52,-13){}  
\node[Nh=7,Nw=7,Nmr=15](V2)(37,0){}  
   \drawedge[ATnb=0,AHnb=0](P3,V1){}  
   \drawedge[ATnb=0,AHnb=0](V1,V2){}  
   \drawedge[ATnb=0,AHnb=0](P2,V1){}  
\drawedge[ATnb=0,AHnb=0](P2,P3){}     
      \drawedge[ATnb=0,AHnb=0](P1,V2){}  
         \drawedge[ATnb=0,AHnb=0](P2,V2){}  
            \drawedge[ATnb=0,AHnb=0](P3,V2){}  
   \end{picture} & $\mv{}{y_{12}}{y_{12}}+\mv{}{y_{23}}{y_{23}}$\\
    \begin{picture}(60,35)(0,0)
\node[Nh=7,Nw=7,Nmr=15](P1)(22,7){} 
\node[Nh=7,Nw=7,Nmr=15](P2)(52,7){} 
\node[Nh=7,Nw=7,Nmr=15](P3)(37,-13){}  
\node[Nh=7,Nw=7,Nmr=15](V1)(52,-13){}  
\node[Nh=7,Nw=7,Nmr=15](V2)(37,0){}  
  \drawedge[ATnb=0,AHnb=0](P1,P2){}  
   \drawedge[ATnb=0,AHnb=0](P3,V1){}  
   \drawedge[ATnb=0,AHnb=0](V1,V2){}  
   \drawedge[ATnb=0,AHnb=0](P2,V1){}  
\drawedge[ATnb=0,AHnb=0](P2,P3){}     
      \drawedge[ATnb=0,AHnb=0](P1,V2){}  
         \drawedge[ATnb=0,AHnb=0](P2,V2){}  
            \drawedge[ATnb=0,AHnb=0](P3,V2){}  
   \end{picture} & $\mv{}{y_{12}}{z}+\mv{}{y_{13}}{z}+\mv{}{y_{23}}{z}$\\
   \begin{picture}(60,35)(0,0)
\node[Nh=7,Nw=7,Nmr=15](P1)(22,7){} 
\node[Nh=7,Nw=7,Nmr=15](P2)(52,7){} 
\node[Nh=7,Nw=7,Nmr=15](P3)(37,-13){}  
\node[Nh=7,Nw=7,Nmr=15](V1)(52,-13){}  
\node[Nh=7,Nw=7,Nmr=15](V2)(37,0){}  
  \drawedge[ATnb=0,AHnb=0](P1,P2){}  
   \drawedge[ATnb=0,AHnb=0](P1,P3){}  
   \drawedge[ATnb=0,AHnb=0](P2,P3){}  
   \drawedge[ATnb=0,AHnb=0](P2,V1){}  
   \drawedge[ATnb=0,AHnb=0](P3,V1){}  
      \drawedge[ATnb=0,AHnb=0](P1,V2){}  
         \drawedge[ATnb=0,AHnb=0](P2,V2){}  
            \drawedge[ATnb=0,AHnb=0](P3,V2){}  
            \drawbpedge[ATnb=0,AHnb=0](P1,20,30,V1,60,40){}
   \end{picture} & $\mv{}{z}{z}$
   \end{tabular}
}
\caption{$\ppt$ frequencies in a general graph.}
\label{table:histop}
\end{table}

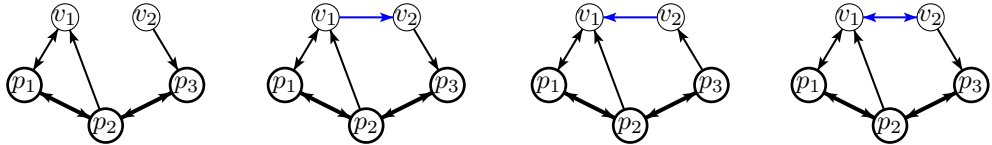
\begin{figure}[ht]
\centerline{
\setlength{\unitlength}{.9mm}
\begin{tabular}{ccccc} 
\hspace{3cm}  &  \hspace{3cm}  & \hspace{3cm} \\ 
 $\ptaux{\Delta}{T_i}{T2}$ &  $\pvaiaux{\Delta}{T_i}{T2}$ &  $\pvoltaaux{\Delta}{T_i}{T2}$ &  $\pvaivoltaaux{\Delta}{T_i}{T2}$\\
~\vspace{1cm}\\
\begin{picture}(0,0)(18,18)
\node[Nh=5,Nw=5,Nmr=4,linewidth=.4](0)(12,15){$p_1$}
\node[Nh=5,Nw=5,Nmr=4,linewidth=.4](1)(24,9){$p_2$}
\node[Nh=5,Nw=5,Nmr=4,linewidth=.4](2)(36,15){$p_3$}
\node[Nh=4,Nw=4,Nmr=4](v1)(18,25){$v_1$}
\node[Nh=4,Nw=4,Nmr=4](v2)(30,25){$v_2$}
\drawedge[ATnb=1,AHnb=1,linewidth=.6,AHLength=2.1,ATLength=2.1](0,1){}
\drawedge[ATnb=1,AHnb=1,linewidth=.6,AHLength=2.1,ATLength=2.1](1,2){}
\drawedge[ATnb=1,AHnb=1,linewidth=.3,AHLength=2.1,ATLength=2.1](0,v1){}
\drawedge[ATnb=0,AHnb=1,linewidth=.3,AHLength=2.1,ATLength=2.1](1,v1){}
\drawedge[ATnb=1,AHnb=0,linewidth=.3,AHLength=2.1,ATLength=2.1](2,v2){}
\end{picture} & 
\begin{picture}(0,0)(18,18)
\node[Nh=5,Nw=5,Nmr=4,linewidth=.4](0)(12,15){$p_1$}
\node[Nh=5,Nw=5,Nmr=4,linewidth=.4](1)(24,9){$p_2$}
\node[Nh=5,Nw=5,Nmr=4,linewidth=.4](2)(36,15){$p_3$}
\node[Nh=4,Nw=4,Nmr=4](v1)(18,25){$v_1$}
\node[Nh=4,Nw=4,Nmr=4](v2)(30,25){$v_2$}
\drawedge[ATnb=1,AHnb=1,linewidth=.6,AHLength=2.1,ATLength=2.1](0,1){}
\drawedge[ATnb=1,AHnb=1,linewidth=.6,AHLength=2.1,ATLength=2.1](1,2){}
\drawedge[ATnb=1,AHnb=1,linewidth=.3,AHLength=2.1,ATLength=2.1](0,v1){}
\drawedge[ATnb=0,AHnb=1,linewidth=.3,AHLength=2.1,ATLength=2.1](1,v1){}
\drawedge[ATnb=1,AHnb=0,linewidth=.3,AHLength=2.1,ATLength=2.1](2,v2){}
\drawedge[ATnb=0,AHnb=1,linewidth=.3,AHLength=2.1,ATLength=2.1,linecolor=blue](v1,v2){}
\end{picture} &
\begin{picture}(0,0)(18,18)
\node[Nh=5,Nw=5,Nmr=4,linewidth=.4](0)(12,15){$p_1$}
\node[Nh=5,Nw=5,Nmr=4,linewidth=.4](1)(24,9){$p_2$}
\node[Nh=5,Nw=5,Nmr=4,linewidth=.4](2)(36,15){$p_3$}
\node[Nh=4,Nw=4,Nmr=4](v1)(18,25){$v_1$}
\node[Nh=4,Nw=4,Nmr=4](v2)(30,25){$v_2$}
\drawedge[ATnb=1,AHnb=1,linewidth=.6,AHLength=2.1,ATLength=2.1](0,1){}
\drawedge[ATnb=1,AHnb=1,linewidth=.6,AHLength=2.1,ATLength=2.1](1,2){}
\drawedge[ATnb=1,AHnb=1,linewidth=.3,AHLength=2.1,ATLength=2.1](0,v1){}
\drawedge[ATnb=0,AHnb=1,linewidth=.3,AHLength=2.1,ATLength=2.1](1,v1){}
\drawedge[ATnb=0,AHnb=1,linewidth=.3,AHLength=2.1,ATLength=2.1](2,v2){}
\drawedge[ATnb=0,AHnb=1,linewidth=.3,AHLength=2.1,ATLength=2.1,linecolor=blue](v2,v1){}
\end{picture} &
\begin{picture}(0,0)(18,18)
\node[Nh=5,Nw=5,Nmr=4,linewidth=.4](0)(12,15){$p_1$}
\node[Nh=5,Nw=5,Nmr=4,linewidth=.4](1)(24,9){$p_2$}
\node[Nh=5,Nw=5,Nmr=4,linewidth=.4](2)(36,15){$p_3$}
\node[Nh=4,Nw=4,Nmr=4](v1)(18,25){$v_1$}
\node[Nh=4,Nw=4,Nmr=4](v2)(30,25){$v_2$}
\drawedge[ATnb=1,AHnb=1,linewidth=.6,AHLength=2.1,ATLength=2.1](0,1){}
\drawedge[ATnb=1,AHnb=1,linewidth=.6,AHLength=2.1,ATLength=2.1](1,2){}
\drawedge[ATnb=1,AHnb=1,linewidth=.3,AHLength=2.1,ATLength=2.1](0,v1){}
\drawedge[ATnb=0,AHnb=1,linewidth=.3,AHLength=2.1,ATLength=2.1](1,v1){}
\drawedge[ATnb=1,AHnb=0,linewidth=.3,AHLength=2.1,ATLength=2.1](2,v2){}
\drawedge[ATnb=1,AHnb=1,linewidth=.3,AHLength=2.1,ATLength=2.1,linecolor=blue](v2,v1){}
\end{picture}
~\\
~\\
~\\
\end{tabular}
}
\caption{Variations of matrix $\ptaux{\Delta}{T_i}{T_j}$ for $T_i=set(A,B,\emptyset)$ and $T_j=set(\emptyset,\emptyset,C)$. }
\label{fig:variationsfive}
\end{figure}

\begin{table}[ht!]
\begin{center}
\begin{tabular}{| l | r | r | r | r |}
\hline
\bf Grafos 	&\bf 	(n,m) 	&\bf $\acc$ (ms) 					& \bf Famod (ms) 		& \bf Kavosh   (ms)\\
\hline
E.coli 	\cite{AlonDataset}		&	$(418, 519)$	& $0.13\pm0.005$ & $8.0\pm0.003$ & $4.5\pm0.02$ \\
CSphd \cite{Pajek2006}		&	$(1882, 1740)$		& $0.68\pm0.009$ & $9.7\pm0.001$ & $5.3\pm0.03$ \\
Yeast~\cite{AlonDataset}	&	$(688, 1079)$& $0.25\pm0.005$ &  $22.7\pm0.03$ & $11.5\pm0.07$ \\
Roget 	\cite{Pajek2006}			&	$(1022, 5074)$		& $2.1\pm0.06$ & $52.7\pm0.01$ & $30.2\pm0.4$ \\
Epa 	 \cite{Pajek2006}			&	$(4271, 8965)$	& $4.4\pm0.08$ & $387.8\pm0.1$ & $202.4\pm0.6$ \\
California \cite{Pajek2006}	&	$(6175, 16150)$	&  $11.3\pm0.14$ & $632.4\pm0.1$
 & $316.3\pm0.5$ \\
Facebook \cite{Tore}				&	$(1899,20296)	$& $14.2\pm0.29$ & $1,446\pm0.6$ & $576\pm2$ \\
ODLIS \cite{Pajek2006}		&	$(2900, 18241)	$	& $18.0\pm0.27$ & $3,150\pm2$ & $957\pm2$ \\
PairsFSG \cite{Pajek2006}	&	$(5018, 63608)$& $105\pm6$ & $3,883\pm14$ & $1,915\pm3$ \\
{\bf Airport	 \cite{Tore}}					&	$\mathbf{(1574,28236)}	$	& $\mathbf{35\pm0.2}$ & $\mathbf{5,772\pm2}$ & $\mathbf{1,250\pm6}$ \\
Foldoc \cite{Pajek2006}		&	$(12905,109092)$& $183\pm0.3$ & $7,481\pm24$ & $2,148\pm 5$ \\
\hline
 \end{tabular} 
\end{center}
\caption{Execution time to count isomorphic patterns of size 3 by processed graph using $\acc$, FANMOD and Kavosh.}
\label{temposk3}
\end{table}

\begin{table}[h!]
\begin{center}
\begin{tabular}{| l | r | r | r | r |}
\hline
\bf Grafos 								&\bf 	(n,m) 				&\bf $\acc$ (ms) 					& \bf Famod (ms) 		& \bf Kavosh   (ms)\\
\hline
E.coli 	\cite{AlonDataset}		&	$(418, 519)$	 & $3.7\pm0.6$ & $221.0\pm0.2$ & $ 	 124.0	 \pm 	 0.7$ \\		
CSphd \cite{Pajek2006}		&	$(1882, 1740)$		& $7.4\pm0.2$ &  $104.5\pm0.3$ &  $ 	 59.1	 \pm 	 0.5$ \\		
Yeast~\cite{AlonDataset}	&	$(688, 1079)	$		& $7.8\pm0.1$ & $573.7\pm0.9$ & $ 	 278.1	 \pm 	 0.8$ \\		
Roget 	\cite{Pajek2006}			&	$(1022, 5074)$		& $58\pm1.3$ & $930.8\pm1.2$ & $ 	 520.1	 \pm 	 2.1$ \\		
Epa 	 \cite{Pajek2006}			&	$(4271, 8965)$		& $200\pm7.4$ & $26,613\pm49$ &  $ 	 13,739	 \pm 	 15$ \\		
California \cite{Pajek2006}	&	$(6175, 16150)$	& $772\pm64$ & $37,484\pm101$ & $ 	 19,280	 \pm 	 110$ \\		
Facebook \cite{Tore}				&	$(1899,20296)	$	& $2135\pm75$ & $147,713\pm169$ & $ 	 59,064	 \pm 	 277$ \\		
{\bf ODLIS \cite{Pajek2006}}		&	$\mathbf{(2900, 18241)}	$	& $\mathbf{2,605\pm115}$ & $\mathbf{630,936\pm2,914}$ &  $ 	 \mathbf{210,015	 \pm 	 837}$ \\		
PairsFSG \cite{Pajek2006}	&	$(5018, 63608)$	& $9,097\pm114$ & $334,025\pm454$ & $ 	 181,957	 \pm 	 617$ \\		
Airport	 \cite{Tore}					&	$(1574,28236)	$	& $10,106\pm130$ & $794,227\pm1091$ & $ 	 163,678	 \pm 	 1,065$ \\		
Foldoc \cite{Pajek2006}		&	$(12905,109092)$&  $10,047\pm30$ & $1,179,759\pm3,066$ & $ 	 259,935	 \pm 	 690$ \\		
\hline
 \end{tabular} 
\end{center}
\caption{Execution time to count isomorphic patterns of size 4 by processed graph using $\acc$, FANMOD and Kavosh.}
\label{temposk4}
\end{table}

\begin{table}[h!]
\begin{center}
\begin{tabular}{| l | r | r | r | r |}
\hline
\bf Grafos 								&\bf 	(n,m) 				&\bf $\acc$ (s)					& \bf FANMOD (s)		& \bf Kavosh (s)  \\
\hline
E.coli 	\cite{AlonDataset}				&	$(418, 519)$						&  $0.1\pm0.01$ 			& $5.6\pm0.02$ 	& $3.5\pm0.03$	\\
CSphd \cite{Pajek2006}				&	$(1882, 1740)$					& $0.1\pm0.002$ 			& $1.3\pm0.03$ 	& $0.76\pm0.01$	\\
Yeast~\cite{AlonDataset}				&	$(688, 1079)	$					&$0.3\pm0.003$ 			& $12.9\pm0.07$ 	& $7.5\pm0.1$\\
Roget 	\cite{Pajek2006}					&	$(1022, 5074)$ 					&  $2\pm0.005$ 			& $17.4\pm0.06$ 	& $11.3\pm0.07$ \\
Epa 	 \cite{Pajek2006}					&	$(4271, 8965)$					& $32.3\pm0.1$ 			& $1,696\pm6$ 		&  $ 	 1,052	 \pm 	 5.9$\\
{\bf California \cite{Pajek2006}}	&	$\mathbf{(6175, 16150)}$ & $\mathbf{76\pm0.4}$ & $\mathbf{2,376\pm12}$ & $\mathbf{1,532\pm 12}$ \\
Facebook \cite{Tore}						&	$(1899,20296)	$ 				&  $505\pm6$ 				& $14,378\pm10$ & $ 	 6,343	 \pm 	 25$\\
ODLIS \cite{Pajek2006}				&	$(2900, 18241)	$ 				&$835\pm2.6$ 				& $>13h$ 				&  $>13h$	\\
PairsFSG \cite{Pajek2006}			&	$(5018, 63608)$				&$2,334\pm2.3$ 			& $>13h$ 				&  $23,036.62 \pm 23.6$	\\
Airport	 \cite{Tore}							&	$(1574,28236)	$				&$3,058\pm14$ 			& $>13h$ 				&  $18,278.96 \pm 37.66$	\\
Foldoc \cite{Pajek2006}				&	$(12905,109092)$			& $2,965\pm3.2$ 			& $>13h$ 				&  $>13h$	\\
\hline
 \end{tabular} 
\end{center}
\caption{Execution time to count isomorphic patterns of size 5 by processed graph using $\acc$, FANMOD and Kavosh.}
\label{temposk5}
\end{table}

%
%

\end{document}